\documentclass{article}

     \PassOptionsToPackage{numbers, compress}{natbib}

\usepackage[preprint]{neurips_2024}

\usepackage[utf8]{inputenc} 
\usepackage[T1]{fontenc}    
\usepackage{hyperref}       
\usepackage{url}            
\usepackage{booktabs}       
\usepackage{amsfonts}       
\usepackage{nicefrac}       
\usepackage{microtype}      
\usepackage{xcolor}         

\usepackage{amsmath,amsthm,amssymb,graphicx,comment,dsfont,fullpage, mathtools, multirow, enumitem, url, indentfirst}
\usepackage{bbm}
\usepackage{wrapfig}
\usepackage{algorithm2e, pdflscape, hyperref, makecell,sidecap}
\RestyleAlgo{ruled}
\usepackage[font=bf]{subfig}
\usepackage{thm-restate}

\AtBeginDocument{}
\theoremstyle{plain}

\newtheorem{proposition}{Proposition}

\theoremstyle{definition}

\usepackage{pdfpages}
\usepackage{adjustbox}
\usepackage{caption}
\usepackage{siunitx} 

\usepackage[capitalise]{cleveref}
\hypersetup{linkcolor=blue, citecolor=blue, colorlinks=true}

\title{Exact Bayesian Gaussian Cox Processes Using Random Integral}

\author{
 Bingjing Tang\\
 Department of Statistics\\
  Stanford University\\
   \texttt{bingjing@stanford.edu}  \\
   \And
   Julia Palacios \\
    Department of Statistics\\
  Stanford University\\
   \texttt{juliapr@stanford.edu} \\
}

\begin{document}

\maketitle

\begin{abstract}\label{sec:abs}
 A Gaussian Cox process is a popular model for point process data, in which the intensity function is a transformation of a Gaussian process. Posterior inference of this intensity function involves an intractable integral (i.e., the cumulative intensity function) in the likelihood resulting in doubly intractable posterior distribution.  Here, we propose a nonparametric Bayesian approach for estimating the  intensity function of an inhomogeneous Poisson process without reliance on large data augmentation or approximations of the likelihood function.  
We propose to jointly model the intensity and the cumulative intensity function as a transformed Gaussian process, 
allowing us to directly bypass the need of approximating the cumulative intensity function in the likelihood. We propose an exact MCMC sampler for posterior inference and evaluate its performance on simulated data. We demonstrate the utility of our method in three real-world scenarios including temporal and spatial event data, as well as aggregated time count data collected at multiple resolutions. 
Finally, we discuss extensions of our proposed method to other 
point processes.
\end{abstract}
\section{Introduction}\label{sec:intro}
Cox processes~\citep{cox1955some,guttorp2012happened} provide useful tools for modeling point process phenomena in a variety of disciplines including biology~\citep{legg2012clustering}, finance~\citep{lando1998cox}, astronomy~\citep{lawrence2016point}, and many others \citep{chen2020neural}. 
Cox processes are also known as doubly stochastic Poisson processes, arising from inhomogeneous Poisson processes with a random intensity measure, which is in turn, a realization from a second random process.

An important class of Cox processes is the log Gaussian Cox processes~\citep{moller1998log}, in which the log intensity function is a Gaussian process ensuring positivity of the intensity function. 
This transformed Gaussian process is a convenient flexible nonparametric prior model of the intensity function without restricting it to a particular functional form. 
Though the class of log Gaussian Cox processes possesses abundant appealing properties and it is widely used as a model for temporal and spatio-temporal point process data~\citep{diggle2013spatial}, the corresponding posterior inference of the intensity function is doubly intractable~\citep{moller2006efficient,murray2012mcmc} due to the intractability of the cumulative intensity function (the integral of the intensity function over data domain) in the likelihood (~\cref{eq:like}). 
A classical approach~\citep{moller1998log,rathbun1994asymptotic,rue2009approximate,taylor2014inla,teng2017bayesian} assumes the intensity function is a finite-dimensional piecewise constant log-Gaussian function. 
Bayesian inference of the intensity function is then performed by MCMC~\citep{lopez2019gaussian,moller1998log,rathbun1994asymptotic,taylor2013lgcp,taylor2015bayesian,teng2017bayesian}, variational Bayes~\citep{teng2017bayesian,watson2024twenty}, or integrated nested Laplace approximation (INLA)~\citep{rue2009approximate,taylor2014inla}. 
Among the three computational approaches, MCMC brings the most accuracy while INLA usually leads to fastest computation~\citep{taylor2014inla,teng2017bayesian}. 
An important aspect in the discretization methods is that the choice of change points (grids) controls the trade-off between numerical accuracy and computational efficiency~\citep{samo2015scalable}.

Apart from those discretization methods, 
nonparametric Bayesian approaches ~\citep{adams2009tractable,kottas2006dirichlet, kottas2007bayesian,  samo2015scalable} have been proposed for exact MCMC or variational inference for Cox processes.
Here `exact' means all terms in the likelihood can be computed without  approximations. \citet{adams2009tractable} eliminates the need for approximating the cumulative intensity function by estimating an augmented posterior distribution of the intensity function and latent points via MCMC.
Specifically, stimulated by the thinning algorithm of~\citet{lewis1979simulation}, a point-process variant of rejection sampling, 
\citet{adams2009tractable} introduced latent (thinned) points in such a way that 
the joint distribution of observed and latent points follows a homogeneous Poisson process.
Therefore the joint posterior inference of 
the intensity function and latent points could be easily implemented via standard MCMC samplers, e.g., a block Gibbs sampler. 
Since this approach does not involve any approximation error in the likelihood, its performance is more accurate than competing discretization methods. 
However, it has several limitations. First, its complexity is cubic in the number of all data points including both observed and latent points. 
Moreover, it is not applicable to high-dimensional input data since the expected number of latent points grows proportionally with the volume of the data domain. For instance, an $n$-dimensional hypercube with sides of length s has volume $s^n$. 

Variational Bayesian approximations for Gaussian Cox processes have been recently proposed~\citep{aglietti2019structured,donner2018efficient,lloyd2015variational}. These methods employ inducing points to reduce computation costs, albeit usually yielding less precise outcomes compared to the aforementioned exact MCMC inference approaches. \citet{donner2018efficient} target an augmented posterior in which in addition to introducing thinning points (as in \citep{adams2009tractable}), the authors add P\'{o}lya-Gamma marks to those points. This allows them to derive analytical expressions of variational posterior mean and covariance. However, their results depend on finite dimensional integrals which cannot be computed analytically. In \citet{lloyd2015variational}, the authors restrict the intensity function to a quadratic function of a Gaussian process.

An alternative approach is to treat the cumulative intensity function as a latent random variable. \citet{kottas2006dirichlet} and \citet{kottas2007bayesian} 
factor 
the intensity function as the product of the normalized intensity function and the cumulative intensity function, each with independent priors. The authors then place a  Beta Dirichlet process mixture prior on 
the normalized intensity function, and a 
Jeffreys prior on the cumulative intensity function. 
Though this model is appealing theoretically, its MCMC posterior computation needs truncating the number of mixture components to some finite number as it is typically done in the Dirichlet process mixture model.
Similarly, \Citet{samo2015scalable} treat the cumulative intensity function as a latent variable with a Gamma prior, however the authors estimate the posterior distribution of the intensity function only at inducing points via MCMC. Motivated by improving scalability, the authors select the set of inducing points by optimizing an utility function in a similar fashion as~\citet{quinonero2005unifying}. 
Instead of placing a functional prior on the intensity function, they construct a finite-dimensional prior on intensity function values at inducing and observed points together with the cumulative intensity function. 

In this paper, in addition to placing a transformed Gaussian process prior on the intensity function, we a priori model the cumulative intensity function as a latent random variable.
We then perform an exact MCMC inference of the augmented posterior distribution of intensity function values at observed and predicted locations together with the cumulative intensity function. Our proposed method involves neither discretization nor  
heavy data augmentation.  Quite different from \citet{kottas2006dirichlet} and \citet{samo2015scalable}, though we also conceive the cumulative intensity function as random, we employ a transformed Gaussian process as the functional prior, which fortunately implies an analytical joint prior over the cumulative intensity function and intensity function values at
a finite set of locations. Moreover, we will show that our proposed method can be further extended to handle count-time data.  To our knowledge, our method is the first exact Gaussian Cox Processes model that includes inference of intensity function from recurrent event data, count data, and a mixture of both types.
More details are presented in the following section. Code implementation of our methods and documentation are available at  
https://github.com/exgpcp/RI.
\section{Methods}
We first assume that we observe a set of $N$ recurrent events $\{s_n\}_{n=1}^N\in\mathcal{S}$ from an inhomogeneous Poisson process on $\mathcal{S} \in \mathbbm{R}^D$ with intensity function $\lambda(\cdot):\mathcal{S}\rightarrow  \mathbbm{R}^+$, and likelihood 
\begin{align}
 p(\{s_n\}_{n=1}^N \vert \lambda)=\exp\left\{ -\Lambda(\mathcal{S}) \right\} \prod_{n=1}^N \lambda(s_n).  \label{eq:like}\end{align}
where $\Lambda(\mathcal{S})=\int_\mathcal{S} \lambda(s)ds$.
Typically, $\lambda(s)$ is a priori modeled as a transformed Gaussian process. A Gaussian process is a stochastic process such that function values at every finite set of locations follow a multivariate Gaussian distribution \citep{williams2006gaussian}. The transformation is chosen to ensure $\lambda(s)>0 \;\forall s\in \mathcal{S}$. Popular choices for transformations are exponential \citep{moller1998log} and sigmoid \citep{adams2009tractable}. An important consequence of this nonparametric prior is that 
the cumulative intensity function $\Lambda(\mathcal{S})$ in \cref{eq:like} becomes intractable since $\lambda(s)$ is an infinite dimensional random function. 
To avoid approximating the cumulative intensity function in \cref{eq:like}, we a priori model $\lambda(s)$ and $\Lambda(\mathcal{S})$ jointly as a positive Gaussian process. In particular,
we a priori model 
$\boldsymbol{\lambda}:=\left[\lambda(x_1),\dots, \lambda(x_M), \Lambda(\mathcal{S})\right]'$ at locations of interest, as a positive truncated multivariate Gaussian distribution conditioned on all its function values being positive. Locations of interest include both observed points $\left\{s_n\right\}_{n=1}^N$, and prediction (test) points $\left\{t_l\right\}_{l=1}^{M-N}$, that is, $\left\{x_i\right\}_{i=1}^M \coloneqq \left\{t_l\right\}_{l=1}^{M-N} \cup\left\{s_n\right\}_{n=1}^N$. 
In addition to event data, we allow data to come as a mixture of recurrent events and binned count data with likelihood
\begin{align} 
 p(\{s_n\}_{n=1}^N, \{c_j\}_{j=1}^J \vert \lambda)=\exp\left\{ -\Lambda(\mathcal{S}) \right\} \prod_{n=1}^N \lambda(s_n) \prod_{j=1}^J\Lambda(B_j)^{c_j}.  \label{eq:like2}\end{align}
 where $\{s_n\}_{n=1}^N\in B_0\subset \mathcal{S}$, $\left\{B_j\right\}_{j=0}^J$ forms a partition of $\mathcal{S}$, and $ c_j$ denotes counts of events in $B_j,\; \forall j \in \left\{1\dots, J\right\}$ . 
\subsection{Positive Gaussian process prior}\label{sec:positive prior}

Our prior on $\boldsymbol{\lambda}$ has many advantages. First, having a positive Gaussian process prior on the intensity function implies that we no longer need to transform it to ensure positivity. This in turn, allows us to a priori model $\Lambda(\mathcal{S})$ (instead of integral over transformation of $\lambda(s)$) as a positive random variable. The consequence is that the vector of the intensity function at a finite set of points and the cumulative intensity function, a priori, follow a well defined positive truncated multivariate Gaussian distribution. The following theorem allows us to couple the intensity function at a finite set of values and the cumulative intensity function, in the unrestricted Gaussian case. The proof can be found in Appendix \ref{appendix:integral}.

\begin{restatable}{theorem}{abc}
Suppose the Gaussian process $f(\cdot)$
on the compact space $\mathcal{X}$  satisfies the assumption that its mean function $\mu(\cdot)$ and covariance kernel $k(\cdot,\cdot)$ are integrable, i.e., $ \int_{\mathcal{X}}\mu(s)ds$,$\int_{\mathcal{X}}k(s,t)dt$ and $\int\int_{\mathcal{X}\times\mathcal{X}}k(s,t)dsdt$ exist.
For every finite set of vectors $s_1, \dots, s_p \in \mathcal{X}$, the corresponding vector $\mathbf{f}\coloneqq \left[f(s_1),\dots, f(s_p), \int_{\mathcal{X}}f(s)ds\right]'$ follows a Gaussian distribution and
\begin{align*}
   \mathbf{f}\sim \mathcal{N}  \left( \boldsymbol{\mu},  \begin{pmatrix}
\boldsymbol{V}_{SS} & \boldsymbol{V}_{SI}\\
\boldsymbol{V}_{SI}' &\boldsymbol{V}_{II} \\
\end{pmatrix} \right),
 \end{align*}
 where $\boldsymbol{\mu}\coloneqq \left[\mu(s_1),\dots, \mu(s_p), \int_\mathcal{X}\mu(s)ds \right]'$, $\boldsymbol{V}_{SS}$ is the $p\times p$ kernel matrix containing covariance terms for all pairs of function values $\left\{f(s_i)\right\}_{i=1}^p$, $\boldsymbol{V}_{SI}$ is a $p$-dimensional vector formed by covariance terms between function values $\left\{f(s_i)\right\}_{i=1}^p$ and $\int_\mathcal{X}f(s)ds$ with $i$-th term being $\int_\mathcal{X}k(s_i, t)dt$, and $\boldsymbol{V}_{II}$ is the covariance of $\int_\mathcal{X}f(s)ds$, whose value is $\int\int _{\mathcal{X}\times\mathcal{X}}k(s,t)dsdt$ .
\label{thm:integral}
\end{restatable}

To define our Gaussian prior, we impose an additional positive constraint to $\boldsymbol{\lambda}$. To be precise, $\boldsymbol{\lambda}$ is a positive Gaussian vector, i.e, $\boldsymbol{\lambda} \sim \mathcal{TN}(\boldsymbol{\mu}, \boldsymbol{V})$, with density:
$$p(\boldsymbol{\lambda})=\frac{\exp\left\{ -\frac{1}{2}(\boldsymbol{\lambda}-\boldsymbol{\mu})'\boldsymbol{V}^{-1} (\boldsymbol{\lambda}-\boldsymbol{\mu})    \right\}}{\int_{\mathcal{F}} \exp\left\{ -\frac{1}{2}(\boldsymbol{\lambda}-\boldsymbol{\mu})'\boldsymbol{V}^{-1} (\boldsymbol{\lambda}-\boldsymbol{\mu})    \right\} d\boldsymbol{\lambda}} \cdot \mathbbm{1}( \boldsymbol{\lambda}> \boldsymbol{0})$$
where $\mathcal{F}=[0,+\infty]^{M+1}$ and $\mathbbm{1}( \boldsymbol{\lambda}> \boldsymbol{0})$ is an indicator function that takes 1 if all elements of $\boldsymbol{\lambda}$ are positive. In this work, we set mean $\boldsymbol{\mu}$ to be zero and 
consider the following two covariance kernels: the squared exponential kernel with hyperparameters  $\theta=(\theta_0,\theta_1)$, i.e., $k_{SE}(x,x')=\theta_0\exp\left({-\frac{\theta_1\|x-x'\|^2}{2}}\right)$, and the Brownian motion covariance kernel with a hyperparameter $\theta$  denoting the precision parameter, i.e., $k_{BM}(x,x')=\frac{1}{\theta}\min(x,x')$ (see their integrals in Appendix \ref{appendix:integralcomp}). We also assume $\mathcal{S}=[0,T]$ and $T$ is known. In general, one can select covariance kernels with analytic integrals, such as linear, squared exponential, and Brownian motion covariance kernels~\citep{williams2006gaussian}.

\subsection{Posterior inference}\label{sec:3.4} 
We are interested in estimating the posterior
\begin{align}
p(\boldsymbol{\lambda}, \theta | \{ x_i\}_{i=1}^M)&\propto\ p_\theta(\theta)\cdot\mathcal{TN}(\boldsymbol{\lambda} ; \boldsymbol{0}, \boldsymbol{V}_\theta) \cdot\exp\left\{ -\Lambda(\mathcal{S}) \right\} \cdot\prod_{n=1}^N \lambda(s_n) ,
\label{eq:truncposterior}
\end{align}
where the covariance $\boldsymbol{V}_\theta$ is constructed from the kernel function $k_\theta(\cdot,\cdot)$ as described in $\cref{thm:integral}$, and the mean of the GP prior is assumed to be zero. 
For the Brownian Motion covariance kernel, we estimate the posterior distribution via Metropolis-within-Gibbs sampling in two steps, alternating between $\boldsymbol{\lambda}$ and $\theta$. For the squared exponential kernel we take an empirical Bayes approach and fix $\theta$ to an estimated value.
\begin{description}
\item[Sample $\boldsymbol{\lambda}\;|\theta,x_1,\dots,x_M$ from:] 
\end{description}
\begin{align}
\begin{split}
p(\boldsymbol{\lambda}\vert \theta, \{ x_i\}_{i=1}^M) 
  \propto \mathcal{N}\left(\boldsymbol{\lambda} ; \mathbf{0},\boldsymbol{V}_\theta\right)\mathbbm{1}\left(\Lambda(\mathcal{S})>0\right)\exp\left\{ -\Lambda(\mathcal{S}) \right\}  \prod_{i=1}^M\mathbbm{1}\left(\lambda(x_i)>0\right) \prod_{n=1}^N \lambda(s_n). \label{eq:pos1}
\end{split}
\end{align}
Sampling from \cref{eq:pos1} via a Metropolis-Hastings algorithm with Gaussian proposal would lead to rare acceptance due to the positivity constraint. However, we found that a routine elliptical slice sampler~\citep{murray2010elliptical} and Hamiltonian Monte Carlo~\citep{neal2011mcmc} work well. 
To run an elliptical slice sampler, we simply pull out the  indicator terms from the truncated Gaussian prior and push them into the likelihood function. 
To perform a Hamiltonian Monte Carlo, the crucial step is computing the gradient of the logarithm of the target density. The gradient and derivation details are in 
Appendix \ref{appendix:gradient}.
In our implementations, we resort to the No-U-Turn HMC sampler~\citep{hoffman2014no}, to avoid time-consuming parameter tuning.

\begin{description}
\item[Sample $\theta|\boldsymbol{\lambda},x_1,\dots,x_M$ from:] 
\end{description}
\begin{align}
p(\theta | \boldsymbol{\lambda},\{ x_i\}_{i=1}^M)&\propto\ p_\theta(\theta)\cdot\frac{\exp\left\{    -\frac{1}{2} \boldsymbol{\lambda}'\boldsymbol{V}_\theta^{-1}\boldsymbol{\lambda} \right\}}{ \int_\mathcal{F} \exp  \left\{-\frac{1}{2} \boldsymbol{\lambda}'\boldsymbol{V}_\theta^{-1}\boldsymbol{\lambda}\right\}d\boldsymbol{\lambda}},\label{eq:pos2}
\end{align}
where $\mathcal{F}=[0,+\infty)^{M+1}$. To sample from~\cref{eq:pos2}, we note that  
the term $\int_\mathcal{F} \exp  \left\{-\frac{1}{2} \boldsymbol{\lambda}'\boldsymbol{V}_\theta^{-1}\boldsymbol{\lambda}\right\}d\boldsymbol{\lambda}$ in~\cref{eq:pos2} needs to be numerically approximated,
therefore simulating from this conditional posterior distribution is a nontrivial problem. For example, HMC sampling from the full conditional of $\theta$  
is not directly applicable here. 
However, in the case of Brownian motion covariance kernel with Gamma prior on $\theta$, we obtain conjugacy and it is possible to sample from \cref{eq:pos2} directly. 
To show this, consider the Brownian motion covariance kernel: $\boldsymbol{V}_\theta=\frac{1}{\theta}\; C$, 
where 
\begin{align}
C=\begin{pmatrix}
x_1&\dots & \min(x_1, x_M)&x_1T-\frac{1}{2}x_1^2 \\
\vdots&\ddots&\vdots&\vdots\\
\min(x_M ,x_1)&\dots &  x_M&x_MT-\frac{1}{2}x_M^2 \\
x_1T-\frac{1}{2}x_1^2&\dots &  x_MT-\frac{1}{2}x_M^2  &\frac{1}{3}T^3 \end{pmatrix}.\label{eq:BMcov}
\end{align}
The marginal posterior now becomes
\begin{align}
p(\theta &| \boldsymbol{\lambda},\{ x_i\}_{i=1}^M)\propto p_\theta(\theta)\cdot\frac{\exp\left\{    -\frac{\theta}{2} \boldsymbol{\lambda}'C^{-1}\boldsymbol{\lambda} \right\}}{ \int_\mathcal{F} \exp  \left\{-\frac{\theta}{2} \boldsymbol{\lambda}'C^{-1}\boldsymbol{\lambda} \right\}d\boldsymbol{\lambda}}\nonumber\\
&=p_\theta(\theta)\cdot\frac{\exp\left\{    -\frac{\theta}{2} \boldsymbol{\lambda}'C^{-1}\boldsymbol{\lambda} \right\}}{ \sqrt{\theta}^{-(M+1)}\int_\mathcal{F} \exp  \left\{-\frac{1}{2} (\sqrt{\theta}\boldsymbol{\lambda})'C^{-1}(\sqrt{\theta}\boldsymbol{\lambda}) \right\}d\sqrt{\theta}\boldsymbol{\lambda}}\nonumber\\
    &=p_\theta(\theta)\cdot\frac{\exp\left\{    -\frac{\theta}{2} \boldsymbol{\lambda}'C^{-1}\boldsymbol{\lambda} \right\}}{ \sqrt{\theta}^{-(M+1)}\int_\mathcal{F} \exp  \left\{-\frac{1}{2} \textbf{z}'C^{-1}\textbf{z} \right\}d\textbf{z}}\propto p_\theta(\theta)\cdot \sqrt{\theta}^{M+1}\exp\left\{    -\frac{\theta}{2} \boldsymbol{\lambda}'C^{-1}\boldsymbol{\lambda} \right\},
\label{eq:posteriortheta}
\end{align}

where $\textbf{z}=\sqrt{\theta}\boldsymbol{\lambda}$. The last step simplifies by noticing that the integral in the denominator no longer depends on $\theta$. Setting $p_\theta(\theta)=Gamma(\alpha,\beta)$, 
\cref{eq:posteriortheta} becomes a Gamma distribution with parameters $\Tilde{\alpha}=\alpha+\frac{M+1}{2}$ and $\Tilde{\beta}=\beta+\frac{1}{2}\boldsymbol{\lambda}'C^{-1}\boldsymbol{\lambda}$. 

Generally, a Brownian motion kernel brings three computational advantages: (1) the integral denominator now becomes analytical with respect to $\theta$, namely, a product of a constant and $\sqrt{\theta}^{-{(M+1)}}$; (2) the conjugacy leads to a parametric marginal posterior distribution amenable for Gibbs sampling; (3) we do not need to compute $\boldsymbol{V}_\theta$, $\boldsymbol{V}_\theta^{-1}$, nor $chol(\boldsymbol{V}_\theta)$ for each iterative update of $\theta$, and instead we just need to compute $C$, $C^{-1}$, and $chol(C)$ only once. Moreover, $C^{-1}$ is a tri-diagonal matrix \citep{rue2005gaussian} amenable to fast sparse matrix computations of matrix inverse and Cholesky decomposition.    

A standard Brownian motion starts from 0 at the origin, i.e. $\lambda(0)=0$, leading to low posterior values around the origin. 
To fix this problem, \citet{rue2005gaussian} proposed to use intrinsic Gaussian
Markov random fields priors with a proper correction at the boundary of the precision matrix $C^{-1}$. This is equivalent to placing a noninformative prior on $\lambda(0)$ and then marginalizing $\lambda(0)$ out. 
 In the rest of this section we will show how to apply this boundary correction technique.
 
First note that 
the distribution of $\boldsymbol{\lambda}$ conditioned on $\lambda(0)=y$ can be expressed as:
\begin{align}    
p\left(\boldsymbol{\lambda}\;|\;\lambda(0)=y,\theta\right)&\propto \exp\left\{ -\frac{\theta}{2} \left( \boldsymbol{\lambda}-y\boldsymbol{l}    \right)'C^{-1}   \left( \boldsymbol{\lambda}-y\boldsymbol{l}     \right)     \right\}\mathbbm{1}(\boldsymbol{\lambda}>\boldsymbol{0}),\label{eq:BMcon}
\end{align}

where $\boldsymbol{l}=(1,\dots,1,T)'$. Even though we conditioned on the initial value at $y$, i.e., $\lambda(0)=y$, the random walk density of $\boldsymbol{\lambda}$ is expressed in terms of function differences $\lambda(x_{i+1})-\lambda(x_{i})=\lambda(x_{i+1})-y-(\lambda(x_{i})-y)$. This allows us to express the conditional density as a multivariate Gaussian distribution with shifted mean. Details are derived in Appendix \ref{appendix:brownian}. We then place a flat Gaussian prior on $\lambda(0)$,  $\mathcal{N}(y;0,\sigma^2)$, with a large value of $\sigma$, and integrate $\lambda(0)$ out to obtain the following  modified random walk density:
\begin{align}
p\left(\boldsymbol{\lambda}\;|\;\theta\right) &= \int_{-\infty}^{+\infty} p\left(\boldsymbol{\lambda}\;|\;\lambda(0)=y,\theta\right)\mathcal{N}(y;0,\sigma^2)dy \approx \exp\left\{ -\frac{\theta}{2} \boldsymbol{\lambda}'\tilde{Q}\boldsymbol{\lambda}\right\}\mathbbm{1}(\boldsymbol{\lambda}>\boldsymbol{0}) \label{eq:BMcorrection}
\end{align}

where $\tilde{Q}=C^{-1}- \cfrac{C^{-1}\boldsymbol{l}\boldsymbol{l}'C^{-1}}{\boldsymbol{l}'C^{-1}\boldsymbol{l}}$. Since $\tilde{Q}$ is rank deficient, we usually add a small perturbation to its diagonal elements to obtain its inverse, i.e., $\tilde{C}=\left( \tilde{Q}+\epsilon\textbf{I}\right)^{-1}$. In our implementations, we simply replace $C$ with $\Tilde{C}$ in \cref{eq:posteriortheta} and obtain
\begin{align}
p(\theta|\boldsymbol{\lambda})&\propto p_\theta(\theta)\cdot \sqrt{\theta}^{M+1}\exp\left\{    -\frac{\theta}{2} \boldsymbol{\lambda}'\Tilde{C}^{-1}\boldsymbol{\lambda} \right\}=\Gamma\left( \alpha+\frac{M+1}{2},\beta+\frac{1}{2}\boldsymbol{\lambda}'\tilde{C}^{-1}\boldsymbol{\lambda}\right)\label{eq:posteriorthetafinal}
\end{align}
To clarify, this intrinsic random walk with precision matrix boundary correction also corresponds to the first-order random walk utilized in the popular integrated nested Laplace approximation model (INLA) \citep{rue2009approximate}, a baseline model we use for comparisons in~\cref{sec:exp}.

For other covariance kernels, we fix the value of $\theta$ to an estimated value first and then we perform posterior sampling of $\boldsymbol{\lambda}$ at locations of interest conditioned on $\theta$ as in \cref{eq:pos1}.

\paragraph{Estimation of $\theta$.} 
We first assume the intensity function is a piecewise constant function according to a regular grid of $m-1$ points. Denote the regular interval length as $\Delta=\frac{T}{m-1}$. The grid points are located at $\left\{\frac{2k-1}{2}\Delta\right\}_{k=1}^{m-1}$ and therefore we have $m$ intervals in total, among which both the first and last intervals have a length of $\frac{\Delta}{2}$. 
 The function values at each interval is denoted by $\boldsymbol{\lambda_m}^*=(\lambda_1^*,\dots,\lambda_m^*)$ and 
 $\Lambda^{*}_{m}$ denotes its corresponding cumulative intensity function. We then optimize the following:
\begin{align}
\operatorname*{arg\,max}_{m,\theta, \boldsymbol{\lambda_m}^*}(1-c)\cdot \log(p\left(s_1,\dots,s_N|\boldsymbol{\lambda_m}^*\right)) +c\cdot \log(\mathcal{TN}\left(\boldsymbol{\lambda_m}^*, \Lambda_m^*;\boldsymbol{0}, \boldsymbol{V}_\theta\right))
\label{eq:map}
\end{align}
where
$c$ is a constant (set to be 0.2 in this work based on empirical observations). 
We name this procedure weighted MAP method. For a fixed value of $m$, we optimize over $\boldsymbol{\lambda_m^*}$ and $\theta$ using R-DEoptim~\citep{mullen2011deoptim}. We repeat this optimization procedure with different choices for $m$ (e.g, $1-10$) and end with the optimal value of \cref{eq:map}.

\paragraph{Extending kernels to multiple dimensions.}\label{sec:multidim}
For point processes observed in higher dimensions, \citet{duvenaud2014automatic} suggests two main kernel constructions, which can be obtained multiplying or adding uni-dimensional kernels. 
For 2-dimesional data $(x_1, x_2), (x_1',x_2')\in \mathbbm{R}^2$, 
the product-kernel is defined as $k_1(x_1,x_1')\times k_2(x_2,x_2')$ and the additive-kernel is defined as $k_1(x_1,x_1')+k_2(x_2,x_2')$. Since a product-kernel offers more flexibility, we use it to model the intensity function for spatial data in the second real example in \cref{sec:realexm3}. 
For a Brownian motion kernel, the product-kernel we use is the two-dimensional Brownian Sheet~\citep{pyke1973partial} with covariance kernel: \begin{align}
k_{BS}((x_1, x_2), (x_1', x_2'))&=\cfrac{1}{\theta}\min(x_1,x_1')\cdot\min(x_2,x_2')\label{eq:BS}
\end{align} 
A squared exponential kernel is directly applicable to high dimensional  input spaces and categorized as a product-kernel.
The hyperparameters estimation and the posterior inference described above for uni-dimensional input space is directly applicable to multi-dimensional input space.

\section{Experiments}\label{sec:exp}
\subsection{Simulations}
We simulated 100 realizations from each of the following two Poisson processes with intensities:
\begin{align}
\begin{split}
     \lambda_1(s)=2\exp\left\{-s/15\right\}+\exp\left\{-((s-25)/10)^2\right\},\,s\in [0, 50]
\end{split}
\quad
\begin{split}
    \lambda_2(s)=10,\,s\in[0, 5].
\end{split}
\label{eq:equations}
\end{align}
For each realization, we compare our proposed random integral (RI) method with the Sigmoid Gaussian Cox Process method (SGCP) proposed by~\citet{adams2009tractable} and the INLA method \citep{rue2009approximate}.
For all the three methods, we assume Gaussian processes priors with a Brownian motion covariance kernel (RI/SGCP/INLA-BM) and place a noninformative conjugate gamma prior on the precision parameter $\theta$ with hyperparameters $\alpha=\beta=0.1$. Additionally, we run both the RI and SGCP methods with Gaussian processes with a squared exponential kernel and estimated hyperparamters with two methods: (1) via the weighted MAP method (RI/SGCP-MAP), as described in~\cref{sec:3.4}; (2) the oracle MLE method (RI/SGCP-MLE) where hyperparameters of the squared exponential kernel are estimated using the true intenstiy function values and observed events. We evaluate the performance of these methods according to sum of squared errors (SSE), coverage, credible intervals width and average running time among each 100 datasets. 
We compute the sum of square errors at grid points by summing up differences between the median of predicted intensity function values and ground truth values.  Coverage at grid points is the percentage of points whose ground truth intensity function values lying between the $95\%$ credible intervals.
Credible interval width is computed as the average width across all grid points. In addition to the simulations under the trajectories of~\cref{eq:equations}, we tested our method with double and triple of the intensities in~\cref{eq:equations}. Corresponding results can be found in~\cref{appendix:addres}. 

Results for one simulated dataset with intensity $\lambda_1(s)$ (\cref{eq:equations}) and 50 number of events, are depicted in the top row of \cref{fig:syn1}, with squared exponential kernels (\cref{fig:syn1}A) and Brownian motion kernels (\cref{fig:syn1}B). Although all methods have comparable performance, INLA-BM (green shaded region in \cref{fig:syn1}B) shows higher uncertainty than any other method. This is consistent across 100 simulations, where INLA's average credible interval width is 1.44 versus 1.20 with our method (see \cref{table:syn1compare}, first 3 rows, last column). Performance statistics based on 100 simulations with intensity $\lambda_1(s)$ shown in \cref{table:syn1compare} indicate that our method is the best performing method among those based on Brownian motion kernels in terms of median SSE. Similarly, for squared exponential kernels with MAP hyperparameter estimation, our method outperforms SGCP in median SSE and coverage. 
Results for one simulated dataset with intensity $\lambda_2(s)$ (\cref{eq:equations}) with 51 events are depicted in the bottom row of \cref{fig:syn1}. Similarly to estimation of $\lambda_1(s)$, INLA-BM displays large uncertainty (\cref{fig:syn1}B).  This again is consistent across 100 simulations for $\lambda_2(s)$ (see \cref{table:syn1compare}, row 8-10, last column.). In terms of median SSE and coverage, our method consistently outperforms others across 100 simulations for $\lambda_2(s)$ (\cref{table:syn1compare}). 
Based on comparisons among the three models over the two synthetic examples, it is concluded that both RI and SGCP outperforms INLA, and RI has a better performance than SGCP for most cases. In this work, methods using Brownian motion kernels display larger uncertainty compared to those using squared exponential kernels. This may be due to the fact that hyperparameters of Brownian motion kernels are estimated within the MCMC procedure, while hyperparameters of squared exponential kernels are fixed with estimates. 

We implement our algorithm in Python and run it on a shared computing cluster consisting of 24 CPU cores and 191 GB of memory and set a time limit of 7 days. For the RI method on temporal data, we run a total of 50,000 iterations after 10,000 iterations of burn-in. For the SGCP model, each iteration took more than 10,000 times longer than the RI method and for this reason, we were only able to obtain about 10,000 iterations after 10,000 iterations of burn-in. However, some results shown in \cref{appendix:addres} of the SGCP method could not be reported. For the second spatial data real example in \cref{sec:realexm3} we run a total of 100,000 iterations after 100,000 iterations of burn-in.
Further, we report average running times among 100 simulated datasets per 10,000 iterations for all MCMC algorithms in both RI and SGCP methods for both synthetic examples in \cref{table:syn1runningtime}. To clarify, for INLA-BM, we report average total running times among 100 simulated datasets.
Though INLA runs fastest, RI has achieved a large improvement over SGCP in terms of running time, which is comparable with INLA. For example, for intensity $\lambda_1(s)$, on average, INLA-BM takes 4.80 s in total and RI-BM takes 15.84 s in total, whereas SGCP-BM takes 13.97 hours in total.   
\begin{figure}
  \centerline{
  \includegraphics[width=0.38\textwidth]{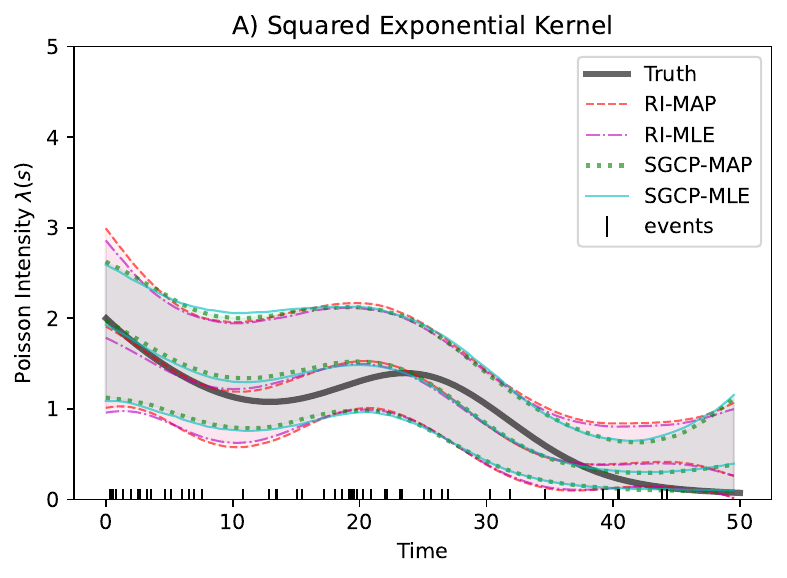}
  \includegraphics[width=0.38\textwidth]{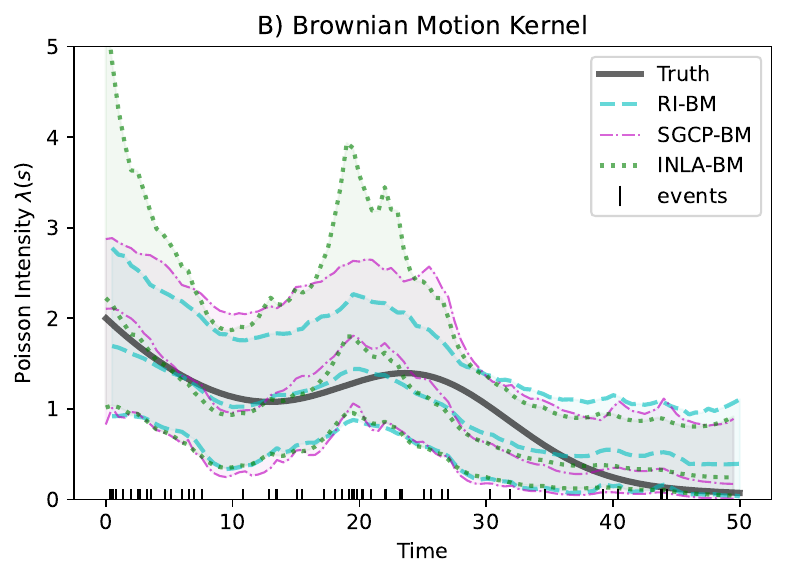}
  }
  \centerline{
  \includegraphics[width=0.38\textwidth]{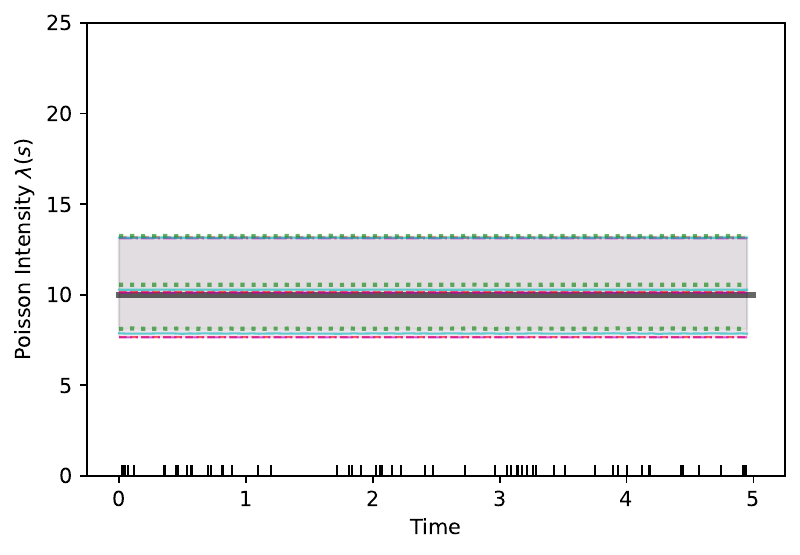}
  \includegraphics[width=0.38\textwidth]{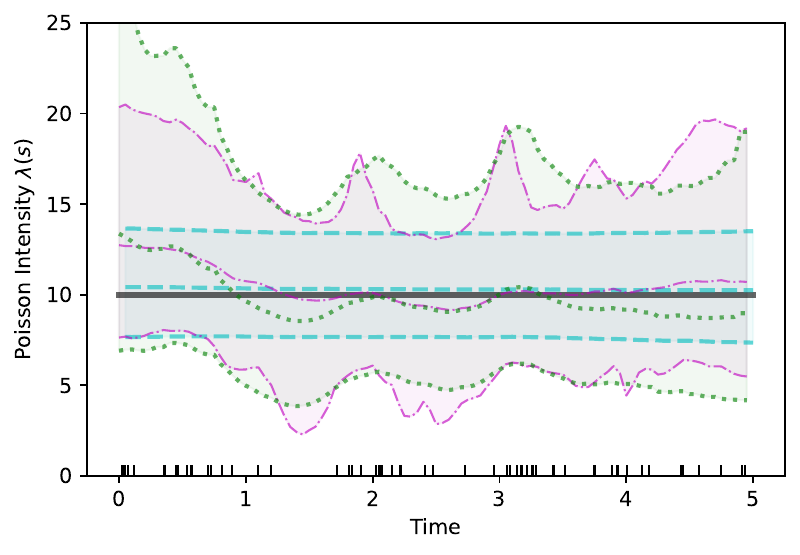}
  }
  \captionsetup{font=footnotesize}
  \vspace{\baselineskip} 
  \caption{
  Top row corresponds to posterior inference of intensity function $\lambda_1(s)$ from one simulated dataset of 50 events, while bottom row corresponds to  posterior inference of intensity function $\lambda_2(s)$ from one simulated dataset of 51 events. True trajectory is depicted by solid black curve, posterior medians by dashed central curves and $95\%$ CI by shaded regions. Time of simulated events are shown as tick marks at the bottom of each plot. Different methods are distinguished by colors as described in legend boxes, and images in the same column share the same lengend box. Methods used in panel A assume squared exponential kernels and those in panel B assume Brownian motion covariance kernels.}
  \label{fig:syn1}
\end{figure}

\begin{table}
\captionsetup{font=footnotesize}
\caption{Performance comparison on simulations with intensities $\lambda_1(s)$ and $\lambda_2(s)$. The last three columns present quantities in the format: 0.50 quantile (0.025 quantile, 0.25 quantile, 0.75 quantile, 0.975 quantile). Bold is the best among methods with the same kernel and red is the best among all methods.}

\vspace{1\baselineskip} 
\label{table:syn1compare}
\centering
\begin{adjustbox}{width=\textwidth} 
\begin{tabular}{p{1.2cm}p{1.8cm}p{5.6cm}p{4.9cm}p{4.6cm}} 
\toprule[1.5pt]
\textbf{Intensity} & \textbf{Methods} & \textbf{SSE at 100 Grids} & \textbf{Coverage at 100 Grids} & \textbf{Credible Interval Width} \\
\midrule
\multirow{7}{*}{$\lambda_1$} & RI-BM & \textbf{7.30} (2.80, 5.22,\;\textcolor{red}{\textbf{9.29}},\; \textcolor{red}{\textbf{14.56}}) & 98\%\;\;\;(\textcolor{red}{\textbf{88\%}}, 95\%, \;\textcolor{red}{\textbf{100\%}}, \textcolor{red}{\textbf{100\%}}) & 1.20 (1.06, 1.15, 1.26, 1.42) \\
& SGCP-BM & 8.87 (\textbf{2.38}, \textbf{5.19}, 12.54, 25.33) & \textcolor{red}{\textbf{100\%}} (82\%, \textcolor{red}{\textbf{100\%}},\textcolor{red}{\textbf{100\%}}, \textcolor{red}{\textbf{100\%}}) & 1.35 (1.11, 1.24, 1.43, 1.71) \\
& INLA-BM & 9.01 (2.92, 6.18, 13.17, 20.94) & \textcolor{red}{\textbf{100\%}} (77\%, 98\%,\; \,\textcolor{red}{\textbf{100\%}}, \textcolor{red}{\textbf{100\%}}) & 1.44 (1.12, 1.31, 1.59, 1.83) \\
\cmidrule(r){2-5}
& RI-MAP & \textbf{7.97} (2.92, \textbf{5.60}, 11.51, 18.45) & \textbf{88\%} \;\;(\textbf{53\%}, \textbf{72\%},\; \textbf{99\%}, \;\textcolor{red}{\textbf{100\%}}) & 0.84 (0.58, 0.71, 1.03, 1.31) \\
& SGCP-MAP & 8.76 (\textbf{2.53}, 6.73, \textbf{11.02}, \textbf{18.31}) & 77\% \;\;\;(41\%, 59\%, \;\;89\%, \;\,\textcolor{red}{\textbf{100\%}}) & 0.84 (0.61, 0.73, 0.99, 1.30) \\
\cmidrule(r){2-5}
& RI-MLE & 7.17 (\textcolor{red}{\textbf{1.42}}, 4.49, \;\textbf{9.62}, \;\textbf{14.76}) & \textcolor{red}{\textbf{100\%}} \;(\textbf{72\%}, \textbf{94\%}, \textcolor{red}{\textbf{100\%}}, \textcolor{red}{\textbf{100\%}}) & 1.04 (0.92, 0.98, 1.07, 1.12) \\
& SGCP-MLE & \textcolor{red}{\textbf{6.33}} (1.61, \textcolor{red}{\textbf{4.14}}, \;9.72,\; 17.30) & \textcolor{red}{\textbf{100\%}} \;(64\%, \;88\%, \textcolor{red}{\textbf{100\%}}, \textcolor{red}{\textbf{100\%}}) & 0.97 (0.85, 0.93, 1.03, 1.09) \\
\midrule[1.5pt]
\multirow{7}{*}{$\lambda_2$} &RI-BM &\textbf{76.63}\;\;\,(\textbf{2.54}, \;\;\textbf{28.65},\; \textbf{220.65},\;\, \textcolor{red}{\textbf{563.71}})& \textcolor{red}{\textbf{100\%}}(\textcolor{red}{\textbf{100\%}},\textcolor{red}{\textbf{100\%}},\textcolor{red}{\textbf{100\%}},\textcolor{red}{\textbf{100\%}})&6.29\;\;\;(5.43,\;\;5.95,\;6.74,\;\;10.18) \\
&SGCP-BM & 231.96  (32.98,  106.33,    429.64,   1720.45)& \textcolor{red}{\textbf{100\%}}(   88\%,\;\;\textcolor{red}{\textbf{100\%}},\textcolor{red}{\textbf{100\%}},\textcolor{red}{\textbf{100\%}}) &9.13\;\;\;(7.79,\;\;8.56,\;10.30,\,13.32) \\
 &INLA-BM &328.58  (74.24,  173.47,     570.52,  1441.48)& \textcolor{red}{\textbf{100\%}}(   95\%,\;\;\textcolor{red}{\textbf{100\%}},\textcolor{red}{\textbf{100\%}},\textcolor{red}{\textbf{100\%}})&12.04 (10.07,11.36,12.90,16.15)\\\cmidrule(r){2-5}
 
 & RI-MAP& \textbf{121.56}  (0.20,   \;\; 36.37, \;\textbf{298.73},\;\textbf{1139.25})
 &\textcolor{red}{\textbf{100\%}}(  \textbf{73\%},\;\;\textcolor{red}{\textbf{100\%}},\textcolor{red}{\textbf{100\%}},\textcolor{red}{\textbf{100\%}})&5.60\;\;\;(4.81,\;\;\;5.26,\;\;5.87,\;\;9.63)\\
 &SGCP-MAP& 145.73  (\textcolor{red}{ \textbf{0.19}}, \;\; \textbf{33.50}, \;    444.92,  1717.75)
 &\textcolor{red}{\textbf{100\%}}(   50\%,\;\;\textcolor{red}{\textbf{100\%}},\textcolor{red}{\textbf{100\%}},\textcolor{red}{\textbf{100\%}})&5.82\;\;\;(4.78,\;\;\;5.43,\;\;6.30,\;11.30)\\\cmidrule(r){2-5}
 & RI-MLE& \textcolor{red}{\textbf{70.82}}\;\;\;(1.14, \;\;\textcolor{red}{\textbf{21.04}}, \;\textcolor{red}{\textbf{216.55}}, \;    854.88)
 &
\textcolor{red}{\textbf{100\%}}(\textcolor{red}{\textbf{100\%}},\;\textcolor{red}{\textbf{100\%}},\textcolor{red}{\textbf{100\%}},\textcolor{red}{\textbf{100\%}})&5.84\;\;\;(4.85,\;\;\;5.41,\;\;6.88,\;14.23)\\
& SGCP-MLE& 105.56\;(\textbf{0.86},\;\;\;37.58,    \;265.23, \;\textbf{826.11})
 &
\textcolor{red}{\textbf{100\%}}(     0\%,\;\;\;\;\textcolor{red}{\textbf{100\%}},\textcolor{red}{\textbf{100\%}},\textcolor{red}{\textbf{100\%}})&5.64\;\;\;(4.77,\;\;\;5.30,\;\;5.95,\;\;6.71)\\
\bottomrule[1.5pt]
\end{tabular}
\end{adjustbox}
\vspace{1\baselineskip} 
\end{table}
\begin{wraptable}{r}{0.5\linewidth} 
    \centering
    \captionsetup{font=footnotesize}
    \caption{Running time for estimating $\lambda_1(s)$ and $\lambda_2(s)$.}
    \label{table:syn1runningtime}
    \vspace{0.5\baselineskip}
    \begin{adjustbox}{max width=\linewidth}
        \begin{tabular}{l>{\centering\arraybackslash}p{3.5cm}>{\centering\arraybackslash}p{3.5cm}}
            \toprule
            \textbf{Methods} & \multicolumn{2}{c}{Average Time $\pm$ Standard Deviation}  \\
            \cmidrule(r){2-3}
            & $\lambda_1$ & $\lambda_2$ \\
            \midrule
            RI-BM & 2.64 $\pm$ 0.11 s  &  2.28 $\pm$ 0.90 s    \\
            SGCP-BM & 28146.07 $\pm$ 4862.67 s & 5724.21 $\pm$ 1025.27 s  \\
            INLA-BM & 4.80 $\pm$ 0.19 s &  4.51 $\pm$ 1.02 s         \\
            \midrule
            RI-MAP & 3.17 $\pm$ 0.44 s  & 3.58 $\pm$ 0.98 s   \\
            SGCP-MAP & 127598.22 $\pm$ 36866.47 s & 41339.5 $\pm$ 12600.63 s     \\
            \midrule
            RI-MLE &3.49 $\pm$ 0.33 s &2.79$\pm$  0.27 s        \\
            SGCP-MLE & 169883.37 $\pm$ 28625.97 s  & 41057.43 $\pm$ 8801.7 s      \\
            \bottomrule
        \end{tabular}
    \end{adjustbox}
\end{wraptable}

\subsection{ Earthquakes in Japan}\label{sec:earthquake}
We are interested in estimating the intensity rate of earthquakes in Japan during the year 2019 with magnitude of at least 2.5. The dataset is gathered from the U.S. Geological Survey (2020)\footnote{from the website https://earthquake.usgs.gov/fdsnws/event/1/}, including 901 time points. Analogous to the synthetic examples, we run the RI model with both a Brownian motion covariance kernel and a squared exponential kernel and present the corresponding results in the top two panels of \cref{fig:real2} (blue dashed lines). We observe that this earthquake temporal point process approximates a homogeneous Poisson process with intensity being 2.50 and that there are more frequent earthquakes during the fourth quarter.

\begin{figure}
  \centerline{
  \includegraphics[width=0.38\textwidth]{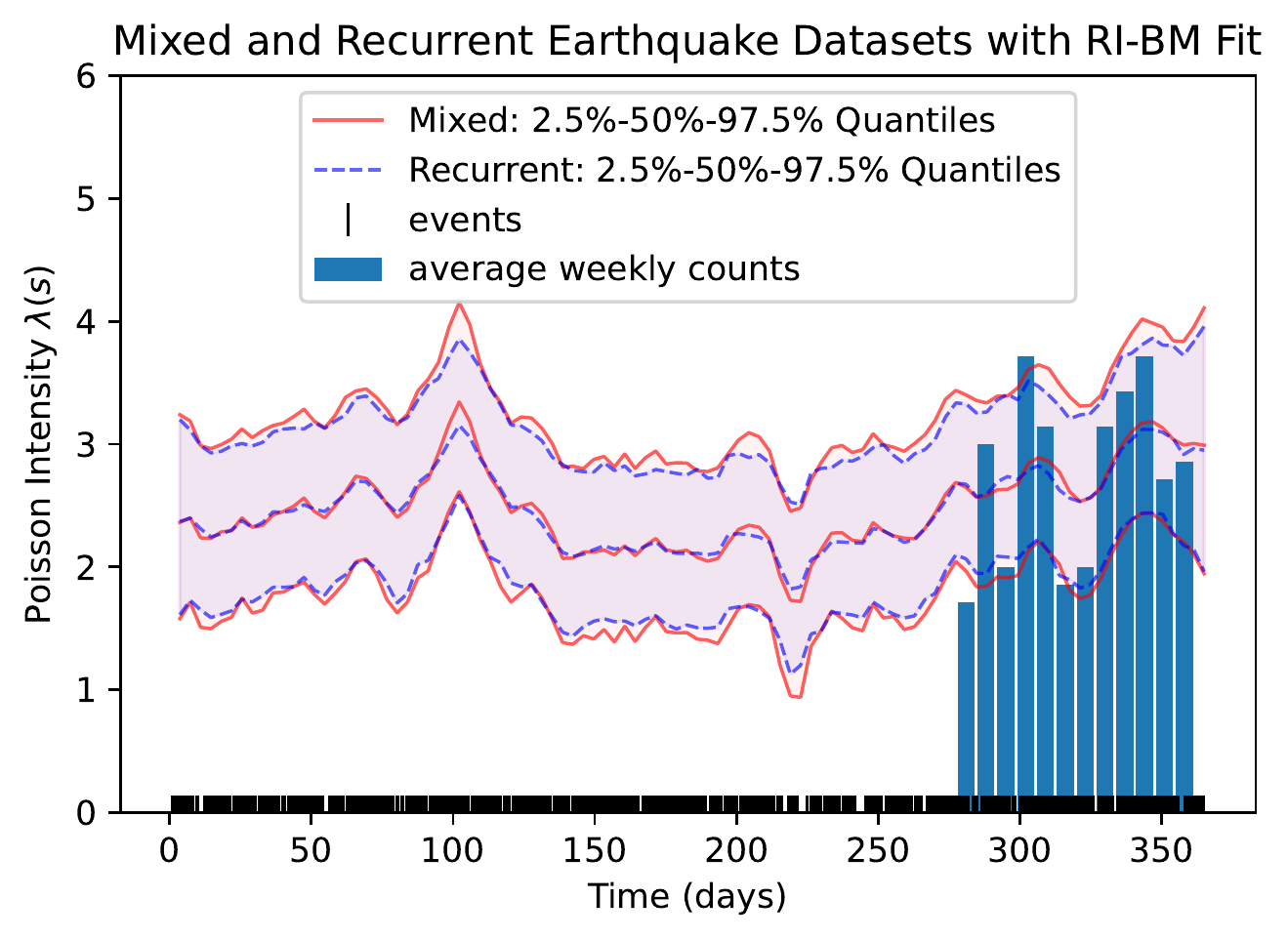}
  \includegraphics[width=0.38\textwidth]{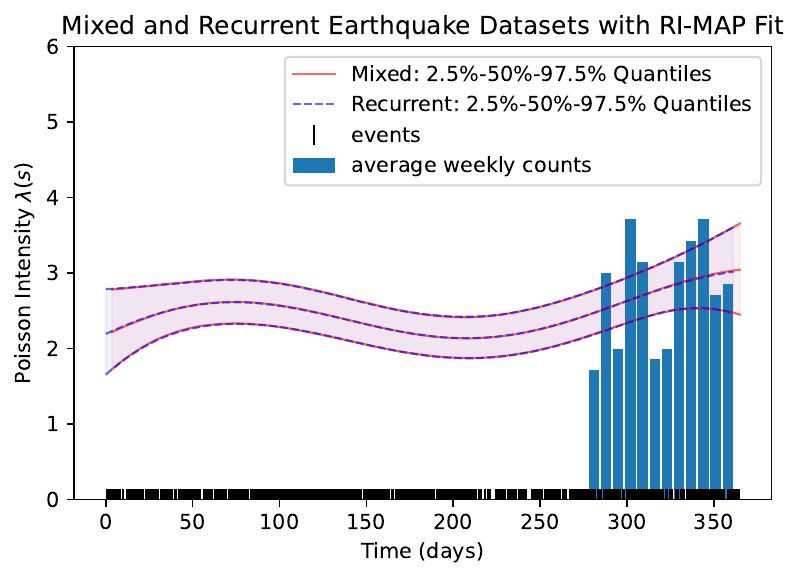}

  }
 \centerline{
  \includegraphics[width=0.38\textwidth]{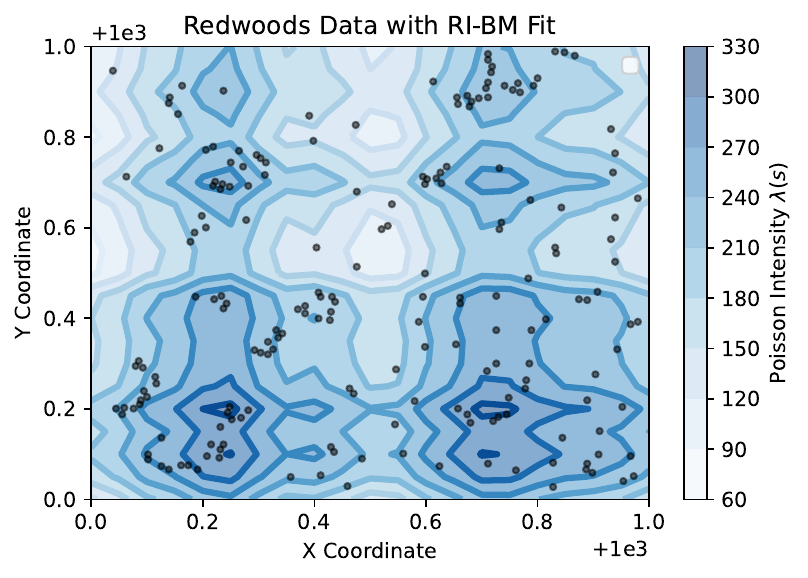}
  \includegraphics[width=0.38\textwidth]{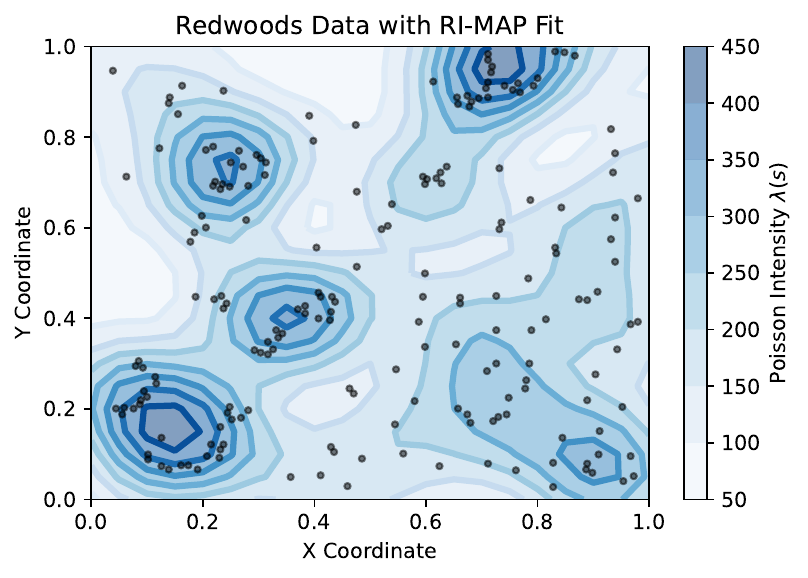}

  }
  \captionsetup{font=footnotesize}
  \vspace{\baselineskip}
  \caption{Posterior inferences for real examples from~\cref{sec:earthquake,sec:realexm3,sec:mixed}.
  Bottom panels report median among posterior samples of the latent GP for the redwoods spatial dataset and locations of observed events are represented by black dots, while top panels present $2.5\%-50\%-97.5\%$ posterior quantiles of the intensity function for both recurrent and mixed earthquake data in Japan and time of observed events are depicted as tick marks at the bottom of each plot. Among each row, the left panel corresponds to results using a Brownian motion covariance kernel and the right panel corresponds to results using a squared exponential kernel with MAP.}\label{fig:real2}
\end{figure}
\subsection{ Redwoods}\label{sec:realexm3}
We reanalyze the redwoods spatial dataset that consists of the locations of seedlings and saplings of California Giant Redwoods in a square sampling region (approximately 130 feet across). It was first described and analysed in 1975 by~\citet{strauss1975model}. The R dataset `redwoodfull' contains the full point pattern of 195 trees and the window has been rescaled to the unit square $[0,1]\times[0,1]$.
Here we use both the Brownian motion covariance kernel and the squared exponential kernel for GP priors on the intensity function. For a Brownian motion covariance kernel, in order to inherit its computation efficiency we resort to the two-dimensional Brownian Sheet with boundary correction as described in \cref{sec:multidim}. To fix the rank deficiency problem of $\tilde{Q}$, we added a small perturbation to its smallest diagonal term. In practice, the original unit square redwoods data requires the perturbation term to be at least $10^{-4}$ in order to obtain a quite accurate estimate of $\tilde{Q}^{-1}$, however leading to a very slow convergent latent GP. Therefore, we place the original unit square redwoods data to the unit square $(1000,1001)\times(1000,1001)$. Results are presented in the bottom two panels of~\cref{fig:real2}. We observe that both methods capture spatial concentration patterns, though there is disagreement in some of the concentrations.

\subsection{ Mixed recurrent and count Data  }\label{sec:mixed}
To exemplify the applicability of our methodology to a mixture of recurrent event and count data, we modified the earthquake data previously analyzed in~\cref{sec:earthquake}. Specifically, we assume the exact recurrent event time points are observed within the first 281 days and the rest of the year (84 days) is reported as weekly counts (12 weeks). In this case, we a priori model the vector $\boldsymbol{\tilde{\lambda}}\coloneqq \left[ \lambda(x_1),\dots, \lambda(x_M), \int_{B_1}\lambda(s)ds,\dots, \int_{B_{12}}\lambda(s)ds 
      \right]'$ as positive multivariate Gaussian. As before, $\left\{x_i\right\}_{i=1}^M$ are the locations of interest, and the $B_j$s correspond to the weekly time intervals (from `00:00:00’ of day $281+7(j-1)$ to `00:00:00’ of day $288+7(j-1)$). Posterior results are also depicted as red curves in the top panel of \cref{fig:real2}. Remarkably, we obtain very similar results to those obtained from the higher resolution recurrent event data (blue curves) analyzed in~\cref{sec:earthquake}.

\section{Discussion}\label{sec:dis}
In this work, we propose an exact nonparametric Bayesian approach to estimate the intensity function of an inhomogeneous Poisson process. The key insight in our method is the realization that if the intensity function is a priori modeled as a Gaussian process, then the Poisson likelihood is a function of a latent Gaussian vector consisting of intensity function values and the cumulative intensity function over the observed set. Since the intensity function needs to be positive, we further restrict the latent vector to being positive. This positive restriction however, does not increase the sampling complexity of our MCMC approach.

The greatest advantage of our method is that it does not require discretization, expensive data augmentation, or variational approximations. Our proposed method is most similar to both \citet{kottas2006dirichlet} and \citet{samo2015scalable}, in that we treat the cumulative intensity function as a latent random variable, and target a posterior distribution augmented by the cumulative intensity function.
Both \citet{kottas2006dirichlet} and \citet{samo2015scalable} reduce time complexity from cubic in the number of both thinning and observed data points in \citet{adams2009tractable} to linear in the number of observed data points for each MCMC iteration. In our proposed method,  attributed to the special structure of Brownian motion kernel covariance, we only need to compute all covariance matrices once during the MCMC procedure, bringing improvement in time efficiency. On the other hand, one limitation of our approach is the requirement of integrable covariance kernels.

Data augmentation methods via Poisson thinning \citep{adams2009tractable,donner2018efficient} require a finite bound and a tractable acceptance probability, restricting its applicability. For example, \citet{teh2011gaussian} mentions the bottleneck of applying the data augmentation via thinning method to renewal processes with unbounded hazard functions. In contrast, the random integral method can handle both unbounded and more complex intractable cases. A future direction is extending the random integral method to more general point processes such as those used in survival analysis~\citep{barrett2013gaussian,fernandez2016gaussian,martino2011approximate}. We anticipate our method will be able to naturally accommodate censored data and spatio-temporally varying covariates.

Finally, we showed that our method can be extended to model mixed data consisting of recurrent event data and count data over binned intervals. Surprisingly, we showed an example in which binning some of the observations resulted in no loss of information about the intensity function. This feature not only increases the applicability of nonparametric Bayesian Cox processes to more realistic scenarios, but it can also be exploited for increasing scalability of our method. We anticipate that it is possible to devise a strategy for binning observations in order to increase the applicability of our method to large datasets.
\newpage
\bibliographystyle{abbrvnat}
\bibliography{biblio}


\appendix
\section{ Appendix: joint distribution of function values and integral for Gaussian processes}\label{appendix:integral}
First, we state a key intermediate result from~\citet{morters2010brownian} that we will use to prove \cref{thm:integral}:
\begin{proposition}{[Proposition 12.15 in \citet{morters2010brownian}]}
Suppose $\left\{X_n : n \in N\right\}$ is a sequence of Gaussian random vectors and $X_n \xrightarrow[]{a.s.} X$.
If $b := lim_{n\rightarrow \infty}E [X_n]$ and $C := lim_{n\rightarrow \infty} Cov [X_n]$ exist,
then $X$ is Gaussian with mean $b$ and covariance matrix $C$.
\label{lm:gaussian rv conv}
\end{proposition}
%
\begin{proof}[\textbf{\upshape Proof:}]
Let $\mu_n$ denote $E [X_n]$, $\Sigma_n$ denote $Cov [X_n]$, and $\psi$ denote the characteristic function. 
Since $X_n \xrightarrow[]{a.s.} X$ implies that $\psi_{X_n}(t)\rightarrow \psi_X(t)$ for each $t$, the rest work is to derive $\lim_{n\rightarrow \infty}\psi_{X_n}(t)$. 
\begin{align*}
    \lim_{n\rightarrow \infty}\psi_{X_n}(t)&=\lim_{n\rightarrow \infty} \exp\left\{it'\mu_n-\frac{1}{2}t'\Sigma_n t\right\} \\
    &=\exp\left\{it' b-\frac{1}{2}t' C t\right\} 
\end{align*}
The last step is due to the continuity of $ \exp\left\{it'\mu_n-\frac{1}{2}t'\Sigma_n t\right\}  $ w.r.t. components of $\mu_n$ and $\Sigma_n$.
Now we have $$\psi_X(t)=\exp\left\{it' b-\frac{1}{2}t' C t\right\}. $$
The result then follows.
\end{proof}

\abc*

\begin{proof}[\textbf{\upshape Proof:}]

It is well known that any affine transformation of a multivariate normal distribution is still a normal distribution. That is to say, if $X \sim \mathcal{N}(X; \mu, \Sigma)$, then
   $B X\sim \mathcal{N}(BX; B\mu, B\Sigma B')$, where $X$ is a $n$-dimensional random vector and $B$ is a $m\times n$ constant matrix. 

\noindent For presentation simplicity, here we only focus on the case when $\mathcal{X}=[0,T]$, i.e., an uni-dimensional space. The proof can be easily extended to a higher finite dimension.
Because $f(s)$ is a continuous function, we can define the integral term $\int_\mathcal{X}f(s)ds $ as the limit of the Riemann sum: $$\int_0^T f(s)ds\coloneq \lim_{n\rightarrow \infty} \frac{T}{n}\sum_{i=1}^n f\left(\frac{i}{n}T\right).$$
Now, define $X_{n}=(f\left(s_1\right),
\ldots,f\left(s_p\right),f\left(\frac{1}{n}T\right),\ldots,
f\left(T\right))’$, then
\begin{align*}
X_n 
\sim \mathcal{N}\left(\tilde{\mu}_n= \begin{pmatrix}
\mu(s_1) \\
\vdots\\
\mu(s_p)\\
\mu(\frac{1}{n}T)\\
\vdots\\
\mu(T)
\end{pmatrix},\tilde{\Sigma}_n=\begin{pmatrix}
k\left(s_1, s_1 \right) &\dots& k\left(s_1, s_p\right)&k\left(s_1, \frac{1}{n}T\right) & \dots &k\left(s_1, T\right)\\
\vdots&\ddots&\vdots&\vdots&\ddots&\vdots\\
k\left(s_p, s_1\right) & \dots&k\left(s_p, s_p\right)  & k\left(s_p, \frac{1}{n}T\right)&\dots& k\left(s_p, T\right) \\
k\left(\frac{1}{n}T, s_1\right)&\dots&k\left(\frac{1}{n}T, s_p\right)  & k\left(\frac{1}{n}T,\frac{1}{n}T\right)&\dots& k\left(\frac{1}{n}T, T\right) \\
\vdots&\ddots&\vdots&\vdots&\ddots&\vdots\\
k\left(T,s_1\right) &\dots& k\left(T, s_p\right)  & k\left(T,\frac{1}{n}T\right)&\dots& k\left(T, T\right) \\
\end{pmatrix} \right) 
\end{align*}

Let $A=\begin{pmatrix}
\textbf{I}_{p\times p}&\mathbf{0}_{p\times n}\\
\mathbf{0}_{1\times p}&\mathbf{\left(\frac{T}{n}\right)}_{1\times n}
\end{pmatrix}$, 
then $AX_n= \begin{pmatrix}
f\left(s_1\right) \\
\vdots\\
f\left(s_p\right)\\
\frac{T}{n}\sum_{i=1}^n f\left(\frac{i}{n}T\right)\\
\end{pmatrix}  \xrightarrow[]{a.s.}  \begin{pmatrix}
f\left(s_1\right) \\
\vdots\\
f\left(s_p\right)\\
\int_0^T f(s) ds\\
\end{pmatrix} $.

Since $AX_n$ is an affine transformation of a multivariate normal distribution, $$AX_n\sim \mathcal{N}(A\tilde{\mu}_n,A\tilde{\Sigma}_nA'),$$

where $A\tilde{\mu}_n=\left( \begin{array}{ccc}
\mu\left(s_1\right) \\
\vdots\\
\mu\left(s_p\right)\\
\frac{T}{n}\sum_{i=1}^n \mu\left(\frac{i}{n}T\right)\\
\end{array} \right) \xrightarrow[]{n\rightarrow \infty}\left( \begin{array}{ccc}
\mu\left(s_1\right) \\
\vdots\\
\mu\left(s_p\right)\\
\int_0^T\mu(s)ds\\
\end{array} \right) $, 

and $A\tilde{\Sigma}_nA'=\begin{pmatrix}
k\left(s_1, s_1 \right) &\dots& k\left(s_1, s_p\right)&\frac{T}{n}\sum_{i=1}^n k\left(s_1, \frac{i}{n}T\right) \\
\vdots&\ddots&\vdots&\vdots\\
k\left(s_p, s_1\right) & \dots&k\left(s_p, s_p\right)  & \frac{T}{n}\sum_{i=1}^n k\left(s_p, \frac{i}{n}T\right) \\
\frac{T}{n}\sum_{i=1}^n k\left( \frac{i}{n}T, s_1\right)&\dots&\frac{T}{n}\sum_{i=1}^n k\left( \frac{i}{n}T, s_p\right)  & \left(\frac{T}{n}\right)^2 \sum_{j=1}^n\sum_{i=1}^n k\left(\frac{i}{n}T,\frac{j}{n}T\right)\\
\end{pmatrix} \xrightarrow[]{n\rightarrow \infty} $

$\begin{pmatrix}
k\left(s_1, s_1 \right) &\dots& k\left(s_1, s_p\right)&\int_0^T k(s_1,t)dt \\
\vdots&\ddots&\vdots&\vdots\\
k\left(s_p, s_1\right) & \dots&k\left(s_p, s_p\right)  & \int_0^T k(s_p,t)dt \\
\int_0^T k(t, s_1)dt&\dots&\int_0^T k(t,s_p)dt  & \int_0^T\int_0^T k(s,t)dsdt\\
\end{pmatrix} $.
\\~\\
The result then follows from \cref{lm:gaussian rv conv}.

\end{proof}

\section{Appendix : kernel integral}\label{appendix:integralcomp}
In this section, we would show how to derive the integral terms in the covariance matrix in \cref{thm:integral} for two specific kernel functions.
\subsection{Squared exponential kernels }
For a Gaussian process with a squared exponential kernel, here we show the derivation for $\int_0^T k_{SE}(s,t)dt$ and $\int_0^T\int_0^T k_{SE}(s,t)ds dt$.
We first claim three facts about the Gauss error function $$erf(z)\coloneq\frac{2}{\sqrt{\pi}}\int_0^z \exp\left\{-t^2\right\}dt.$$
\paragraph{FACT 1:} $\int_p^q\exp\left\{-\frac{x^2}{2}\right\} dx=\sqrt{\frac{\pi}{2}}\left(erf\left(q/\sqrt{2}\right)-erf\left(p/\sqrt{2}\right)\right)$
\paragraph{FACT 2:}$\int erf(z)dz=z\, erf(z)+\frac{\exp\left\{-z^2\right\}}{\sqrt{\pi}}+c$
\paragraph{FACT 3:} $erf(-z)=-erf(z)$

Utilizing the three facts, we compute both single and double integral terms as below. 
\begin{align*}
    &\int_0^T k_{SE}(s,t)dt\\
    &=\int_0^T\theta_0\exp\left({-\frac{\theta_1(s-t)^2}{2}}\right)dt\\
    &=\frac{\theta_0}{\sqrt{\theta_1}}\int_0^T\exp\left({-\frac{\theta_1(t-s)^2}{2}}\right)  d\sqrt{\theta_1}(t-s)\\
    &=\frac{\theta_0}{\sqrt{\theta_1}}\int_{-\sqrt{\theta_1}s}^{\sqrt{\theta_1}(T-s)} \exp\left\{-\frac{v^2}{2}\right\} dv\\
    &=\frac{\theta_0}{\sqrt{\theta_1}}\sqrt{\frac{\pi}{2}}\left[erf\left( \sqrt{\frac{\theta_1}{2}}\left(T-s\right)\right)-erf\left(-\sqrt{\frac{\theta_1}{2}}s \right)\right]\\
    &=\theta_0\sqrt{\frac{\pi}{2\theta_1}}\left[erf\left(\sqrt{\frac{\theta_1}{2}}\left(T-s\right)\right)+erf\left(\sqrt{\frac{\theta_1}{2}}s\right)\right]
\end{align*}
where $v\coloneq\sqrt{\theta_1}(t-s)$.
\begin{align*}
&\int_0^T\int_0^T k_{SE}(s,t)dt ds\\&=\int_0^T \theta_0\sqrt{\frac{\pi}{2\theta_1}}\left[erf\left(\sqrt{\frac{\theta_1}{2}}\left(T-s\right)\right)+erf\left(\sqrt{\frac{\theta_1}{2}}s\right)\right]ds\\
&=\theta_0\sqrt{\frac{\pi}{2\theta_1}}\left[\int_0^T erf\left(\sqrt{\frac{\theta_1}{2}}\left(T-s\right)\right) ds + \int_0^T erf\left(\sqrt{\frac{\theta_1}{2}}s\right) ds\right]\\
&=\theta_0\sqrt{\frac{\pi}{2\theta_1}}\left[-\sqrt{\frac{2}{\theta_1}}\int_0^T erf\left(\sqrt{\frac{\theta_1}{2}}\left(T-s\right)\right) d\sqrt{\frac{\theta_1}{2}}(T-s) + \sqrt{\frac{2}{\theta_1}}\int_0^T erf\left(\sqrt{\frac{\theta_1}{2}}s\right) d \sqrt{\frac{\theta_1}{2}}s\right]\\
&=\theta_0\sqrt{\frac{\pi}{2\theta_1}}\left[-\sqrt{\frac{2}{\theta_1}}\int_0^T erf\left(\sqrt{\frac{\theta_1}{2}}\left(T-s\right)\right) d\sqrt{\frac{\theta_1}{2}}(T-s) + \sqrt{\frac{2}{\theta_1}}\int_0^T erf\left(\sqrt{\frac{\theta_1}{2}}s\right) d \sqrt{\frac{\theta_1}{2}}s\right]\\
&=\frac{\sqrt{\pi}\theta_0}{\theta_1}\left[ \int_0^{\sqrt{\frac{\theta_1}{2}}T} erf\left(u \right)du+ \int_0^{\sqrt{\frac{\theta_1}{2}}T} erf\left(v\right) dv   \right]\\
&=\frac{2\sqrt{\pi}\theta_0}{\theta_1}\left[u\;erf(u)+\frac{\exp\left\{-u^2 \right\}}{\sqrt{\pi}} \Bigg\vert_0^{\sqrt{\frac{\theta_1}{2}}T}\right]\\
&=\frac{2\sqrt{\pi}\theta_0}{\theta_1}\left[\sqrt{\frac{\theta_1}{2}}T\;erf\left(\sqrt{\frac{\theta_1}{2}}T\right)+\frac{\exp\left\{-\frac{\theta_1}{2} T^2 \right\}}{\sqrt{\pi}} -\frac{1}{\sqrt{\pi}}\right]\\
&=\frac{2\theta_0}{\theta_1}\left[\sqrt{\frac{\pi\theta_1}{2}}T\;erf\left(\sqrt{\frac{\theta_1}{2}}T\right)+\exp\left\{-\frac{\theta_1}{2} T^2 \right\} -1\right]
\end{align*}
where $u\coloneq\sqrt{\frac{\theta_1}{2}}(T-s)$ and $v\coloneq\sqrt{\frac{\theta_1}{2}}s$.
\subsection{Brownian motion covariance kernels}
Next, we would show a similar derivation for a Gaussian process with a Brownian motion covariance kernel.
\begin{align*}
    \int_0^T k_{BM}(s,t)dt&=\int_0^T \frac{1}{\theta}\min\left(s,t\right)dt\\
    &=\frac{1}{\theta}\left[\int_0^s t dt+\int_s^T s dt\right]\\
    &=\frac{1}{\theta}\left[\frac{1}{2}t^2\Big\vert_0^s+st\Big\vert_s^T\right]\\
    &=\frac{1}{\theta}\left(sT-\frac{1}{2}s^2\right)
\end{align*}
\begin{align*}
    \int_0^T\int_0^T k_{BM}(s,t)dt ds&=\int_0^T\frac{1}{\theta}\left(sT-\frac{1}{2}s^2\right)ds\\
    &=\frac{1}{\theta}\left( \frac{1}{2}Ts^2-\frac{1}{6}s^3 \right)\Bigg\vert_0^T\\
    &=\frac{1}{3\theta}T^3
\end{align*}
\section{Appendix : gradient of logarithm of the conditional posterior density}\label{appendix:gradient}
\noindent Derived from the conditional posterior density in \cref{eq:pos1}, we could easily write its logarithm as below, which could be splitted into 3 terms $g_1$, $g_2$, and $g_3$:
\begin{align}
&\log p(\boldsymbol{\lambda} \;| \{ x_i\}_{i=1}^M,\theta)\nonumber\\
&=\log\mathcal{N}\left(\boldsymbol{\lambda} ; \mathbf{0},\boldsymbol{V}_\theta\right)+\log\mathbbm{1}\left(\int_\mathcal{S}\lambda(s)ds>0\right) +\sum_{j=1}^{M-N} \log \mathbbm{1}\left(\lambda(t_j)>0\right)+\sum_{n=1}^N \log \mathbbm{1}\left(\lambda(s_n)>0\right)\nonumber\\
&-\int_\mathcal{S} \lambda(s) ds+ \sum_{n=1}^N \log \lambda(s_n) \nonumber\\
&\coloneq g_1\left(\boldsymbol{\lambda}\right)+g_2\left(\boldsymbol{\lambda}\right)+g_3\left(\lambda(s_1),\dots,\lambda(s_N),\int_\mathcal{S}\lambda(s)ds\right)\nonumber
\end{align}
where 
\begin{align*}
g_1\left(\boldsymbol{\lambda}\right)&\coloneq \log\mathcal{N}\left(\boldsymbol{\lambda}
 ; \mathbf{0},\boldsymbol{V}_\theta\right)\nonumber\\
&=-\frac{\boldsymbol{\lambda}' \boldsymbol{V}_\theta^{-1}\boldsymbol{\lambda}}{2} -\frac{M+1}{2}\log 2\pi -\frac{1}{2}\log|\boldsymbol{V}_\theta|
\end{align*}

\begin{align*}
g_2\left(\boldsymbol{\lambda}\right)&\coloneq \log\mathbbm{1}\left(\int_\mathcal{S}\lambda(s)ds>0\right) +\sum_{j=1}^{M-N} \log \mathbbm{1}\left(\lambda(t_j)>0\right)+\sum_{n=1}^N \log \mathbbm{1}\left(\lambda(s_n)>0\right)
\end{align*}

\begin{align*}
g_3\left(\lambda(s_1),\dots,\lambda(s_N),\int_\mathcal{S}\lambda(s)ds\right)\coloneq -\int_\mathcal{S} \lambda(s) ds+ \sum_{n=1}^N \log \lambda(s_n)
\end{align*}

\noindent To compute the gradient for $g_1$, we could apply the fact of matrix calculus that, $\frac{\partial X'A X}{\partial X}=2A X$, where $X $ is a $p$-dimensional vector and $A$ is a $p\times p$ matrix, to the derivative of logarithm of the normal distribution.
\begin{align*}
\nabla g_1\left(\boldsymbol{\lambda} \right)&=\left(\frac{\partial g_1}{\partial \lambda(x_1)}, \dots, \frac{\partial g_1 }{\partial \lambda(x_M)},\frac{\partial g_1}{\partial \int_\mathcal{S} \lambda(s)ds}  \right)' \nonumber\\
&=-\boldsymbol{V}_\theta^{-1}\boldsymbol{\lambda}
\end{align*}
Next, we aim to compute the derivative $g_2$ including all the indicator terms, at any nonzero location.
\begin{align*}
\nabla g_2\left(\boldsymbol{\lambda}\right)&=\left(\frac{\partial g_2}{\partial \lambda(x_1)}, \dots, \frac{\partial g_2 }{\partial \lambda(x_M)}, \frac{\partial g_2}{\partial \int_\mathcal{S} \lambda(s)ds}  \right)' \nonumber\\
&=(0,\dots,0,0)'
\end{align*}
The rest work is to compute the derivative of $g_3$.
Since the integral term could be written as a limit format $\int_\mathcal{S} \lambda(s)ds=lim_{\Delta x\rightarrow 0}\Delta x \cdot \sum_i \lambda(x_i)$, its derivative w.r.t any $\lambda(x_i)$ is always $\lim_{\Delta x\rightarrow 0}\Delta x=0$. Hence we have the below:
\begin{align*}
\nabla g_3\left(\lambda(s_1),\dots, \lambda(s_N),\int_\mathcal{S} \lambda(s)ds \right)&=\left(\frac{\partial g_3}{\partial \lambda(t_1)}, \dots, \frac{\partial g_3 }{\partial \lambda(t_{M-N})},\frac{\partial g_3}{\partial \lambda(s_1)}, \dots, \frac{\partial g_3 }{\partial \lambda(s_N)}, \frac{\partial g_3}{\partial \int_\mathcal{S} \lambda(s)ds}  \right)' \nonumber\\
&=\left(0,\dots,0,\frac{1}{\lambda(s_1)},\dots,\frac{1}{\lambda(s_N)},-1 \right)'
\end{align*}   
Summing up the derivatives of the three terms, we could derive the derivative of logarithm of the conditional posterior distribution as below.
\begin{align*}
&\nabla \log p(\boldsymbol{\lambda}\; | \{ x_i\}_{i=1}^M,\theta) =\left(\frac{\partial \log p}{\partial \lambda(t_1)},\dots,\frac{\partial \log p}{\partial \lambda(t_{M-N})}, \frac{\partial \log p}{\partial \lambda(s_1)},\dots,\frac{\partial \log p}{\partial \lambda(s_N)},\frac{\partial \log p}{\partial \int_{\mathcal{S}}\lambda(s)ds }\right)' \nonumber\\
&=\left(0,\dots,0, \frac{1}{\lambda(s_1)},\dots,\frac{1}{\lambda(s_N)},-1 \right)'-\boldsymbol{V}_\theta^{-1}\boldsymbol{\lambda}
\end{align*}
\section{Appendix : Brownian motion precision matrix boundary correction}\label{appendix:brownian}
\noindent We first show how to derive $p\left(\boldsymbol{\lambda}\;|\;\lambda(0)=y,\theta\right)$ in \cref{eq:BMcon}. For a Brownian motion starting at $y$, the conditional distribution of $\left(\lambda(x_1),\dots, \lambda(x_M)\right)$ given $\lambda(0)=y$ is:
\begin{align}    
p(\lambda(x_1),\dots, \lambda(x_M)\vert\lambda(0)=y,\theta)&\propto \exp\left\{-\frac{\theta}{2}\left[\frac{ \left(\lambda(x_{(1)})-y \right)^2   }{x_{(1)}}    +\sum_{i=2}^{M}\frac{\left[\lambda(x_{(i)})-\lambda(x_{(i-1)})\right]^2}{x_{(i)}-x_{(i-1)}}\right]   \right\}\nonumber\\
&= \exp\left\{ -\frac{\theta}{2} \begin{pmatrix}
\lambda(x_{(1)})-y\\ \vdots\\  \lambda(x_{(M}))-y \end{pmatrix}'Q \begin{pmatrix}
\lambda(x_{(1)})-y\\ \vdots\\  \lambda(x_{(M)})-y \end{pmatrix}      \right\}\nonumber\\
&= \exp\left\{ -\frac{\theta}{2} \begin{pmatrix}
\lambda(x_1)-y\\ \vdots\\  \lambda(x_M)-y\\ 
 \end{pmatrix}'\Sigma^{-1}  \begin{pmatrix}
\lambda(x_1)-y\\ \vdots\\  \lambda(x_M)-y\\ 
 \end{pmatrix}      \right\}
\label{eq:condBM}
\end{align}
where $x_{(1)},\dots,x_{(M)}$ are the increasingly ordered locations of  $x_1,\dots, x_M$, $Q$ is a symmetric matrix with components $
Q_{ij} = 
     \begin{cases}
       \cfrac{1}{x_{(1)}}+\cfrac{1}{x_{(2)}-x_{(1)}} &\quad\text{if}\; i=j=1\\
       \cfrac{1}{x_{(i+1)}-x_{(i)}}+\cfrac{1}{x_{(i)}-x_{(i-1)}} &\quad\text{if}\; 1<i=j<M\\
        \cfrac{1}{x_{(M)}-x_{(M-1)}} &\quad\text{if}\; i=j=M\\
         -\cfrac{1}{x_{(i+1)}-x_{(i)}} &\quad\text{if}\; j=i+1\\
0 &\quad\text{otherwise} \\
     \end{cases}
$,

and $\Sigma=\begin{pmatrix}
x_1 & \dots&\min(x_1,x_M) \\
\vdots&\ddots&\vdots\\
\min(x_M,x_1) & \dots&x_M \\
\end{pmatrix}$.
\\
Derived from \cref{eq:condBM} and \cref{thm:integral}, we obtain $p\left(\boldsymbol{\lambda}\;|\;\lambda(0)=y,\theta\right)$ as below.
\begin{align*}
p\left(\boldsymbol{\lambda}\;|\;\lambda(0)=y,\theta\right)
&\propto \exp\left\{ -\frac{\theta}{2} \begin{pmatrix}
\lambda(x_1)-y\\ \vdots\\  \lambda(x_M)-y\\ 
\int_0^T \lambda(s)ds-Ty\end{pmatrix}'C^{-1}  \begin{pmatrix}
\lambda(x_1)-y\\ \vdots\\  \lambda(x_M)-y\\ 
\int_0^T \lambda(s)ds-Ty\end{pmatrix}      \right\}\\
&=\exp\left\{ -\frac{\theta}{2} \left(\boldsymbol{\lambda}-y\boldsymbol{l}\right)'C^{-1}  \left(\boldsymbol{\lambda}-y\boldsymbol{l}\right) \right\}
\end{align*}
where $\boldsymbol{l}=(1,\dots,1,T)'$ and $C$ is the same as~\cref{eq:BMcov}.

Next, we will show how to derive the distribution of a positive truncated random walk, after boundary correction, corresponding to \cref{eq:BMcorrection}.
\begin{align*}
&p\left(\boldsymbol{\lambda} \;|\; \theta\right)\nonumber\\
&=\int_{-\infty}^{+\infty} p\left(\boldsymbol{\lambda}\;|\;\lambda(0)=y,\theta\right)\mathcal{N}(y;0,\sigma^2)dy \nonumber\\
&\propto\mathbbm{1}(\boldsymbol{\lambda}>\boldsymbol{0})\int_{-\infty}^{+\infty}  \exp\left\{ -\frac{\theta}{2} \left[\left(y\boldsymbol{l}-\boldsymbol{\lambda}\right)'C^{-1}  \left(y\boldsymbol{l}-\boldsymbol{\lambda}\right) +\frac{y^2}{\theta\sigma^2}     \right]\right\}dy\nonumber\\
&=\mathbbm{1}(\boldsymbol{\lambda}>\boldsymbol{0})\int_{-\infty}^{+\infty}  \exp\left\{ -\frac{\theta}{2} \left[     y^2\left(\boldsymbol{l}'C^{-1}\boldsymbol{l}+\frac{1}{\theta\sigma^2}\right) -2y\boldsymbol{l}'C^{-1}\boldsymbol{\lambda}+\boldsymbol{\lambda}'C^{-1}\boldsymbol{\lambda}   \right]\right\}dy\nonumber\\
&=\mathbbm{1}(\boldsymbol{\lambda}>\boldsymbol{0})\int_{-\infty}^{+\infty}  \exp\left\{ -\frac{\theta}{2} \left[ \left(\boldsymbol{l}'C^{-1}\boldsymbol{l}+\frac{1}{\theta\sigma^2}\right)\left(y-\frac{\boldsymbol{l}'C^{-1}\boldsymbol{\lambda}}{\boldsymbol{l}'C^{-1}\boldsymbol{l}+\frac{1}{\theta\sigma^2}} \right)^2 +\boldsymbol{\lambda}'C^{-1}\boldsymbol{\lambda} -\frac{(\boldsymbol{l}'C^{-1}\boldsymbol{\lambda})^2}{\boldsymbol{l}'C^{-1}\boldsymbol{l}+\frac{1}{\theta\sigma^2}}    \right]\right\}dy\nonumber\\
&= \mathbbm{1}(\boldsymbol{\lambda}>\boldsymbol{0})\exp\left\{ -\frac{\theta}{2} \left[ \boldsymbol{\lambda}'C^{-1}\boldsymbol{\lambda} -\frac{(\boldsymbol{l}'C^{-1}\boldsymbol{\lambda})^2}{\boldsymbol{l}'C^{-1}\boldsymbol{l}+\frac{1}{\theta\sigma^2}} \right] \right\}\int_{-\infty}^{+\infty}  \exp\left\{ -\frac{\theta\left(\boldsymbol{l}'C^{-1}\boldsymbol{l}+\frac{1}{\theta\sigma^2}\right)}{2} \left(y-\frac{\boldsymbol{l}'C^{-1}\boldsymbol{\lambda}}{\boldsymbol{l}'C^{-1}\boldsymbol{l}+\frac{1}{\theta\sigma^2}} \right)^2   \right\}dy       \nonumber\\
&= \mathbbm{1}(\boldsymbol{\lambda}>\boldsymbol{0})\exp\left\{ -\frac{\theta}{2} \left[ \boldsymbol{\lambda}'C^{-1}\boldsymbol{\lambda} -\frac{(\boldsymbol{l}'C^{-1}\boldsymbol{\lambda})^2}{\boldsymbol{l}'C^{-1}\boldsymbol{l}+\frac{1}{\theta\sigma^2}} \right] \right\} \left(\theta\boldsymbol{l}'C^{-1}\boldsymbol{l}+\frac{1}{\sigma^2}\right)^{-\frac{1}{2}}         \int_{-\infty}^{+\infty}  \exp\left\{ -\frac{z^2}{2} \right\}dz     \nonumber\\
&\propto \mathbbm{1}(\boldsymbol{\lambda}>\boldsymbol{0})\exp\left\{ -\frac{\theta}{2} \left[ \boldsymbol{\lambda}'C^{-1}\boldsymbol{\lambda} -\frac{(\boldsymbol{l}'C^{-1}\boldsymbol{\lambda})^2}{\boldsymbol{l}'C^{-1}\boldsymbol{l}+\frac{1}{\theta\sigma^2}} \right] \right\}\nonumber\\
&=\mathbbm{1}(\boldsymbol{\lambda}>\boldsymbol{0})\exp\left\{ -\frac{\theta}{2} \boldsymbol{\lambda}'\left[ C^{-1}- \frac{C^{-1}\boldsymbol{l}\boldsymbol{l}'C^{-1}}{\boldsymbol{l}'C^{-1}\boldsymbol{l}+\frac{1}{\theta\sigma^2}}   \right]\boldsymbol{\lambda}\right\}\nonumber\\
&\colonapprox \mathbbm{1}(\boldsymbol{\lambda}>\boldsymbol{0})\exp\left\{ -\frac{\theta}{2} \boldsymbol{\lambda}'\tilde{Q}\boldsymbol{\lambda}\right\}
\end{align*}
where $z\coloneq\sqrt{\theta\boldsymbol{l}'C^{-1}\boldsymbol{l}+\frac{1}{\sigma^2}}\left(y-\cfrac{\boldsymbol{l}'C^{-1}\boldsymbol{\lambda}}{\boldsymbol{l}'C^{-1}\boldsymbol{l}+\frac{1}{\theta\sigma^2}}\right)$, and $\tilde{Q}\coloneq C^{-1}- \cfrac{C^{-1}\boldsymbol{l}\boldsymbol{l}'C^{-1}}{\boldsymbol{l}'C^{-1}\boldsymbol{l}} $.

\section{Appendix : additional experiments results}\label{appendix:addres}
\subsection{Simulations}
Here we present extra qualitative and quantitative experiment results for the two Poisson processes with intensities in~\cref{eq:equations}. In addition to~\cref{fig:syn1}, ~\cref{fig:syn1traceplots,fig:syn2traceplots} present traceplots of the latent GP at midpoint and histograms of the cumulative intensity function for the two simulated datasets with intensities $\lambda_1(s)$ and $\lambda_2(s)$ (\cref{eq:equations}). All results show that the MCMC algorithm converges well and modes of histograms approach the ground truth values for the cumulative intensity function.~\cref{table:syn1full,table:syn1timefull,table:syn2full,table:syn2timefull} present quantiles of the same performance evaluation metrics as those reported in \cref{table:syn1compare,table:syn1runningtime} across each simulated 100 datasets from the original, double and
triple of both intensities. In addition to reporting those quantities at grid points, we also report them at observed locations.
\begin{figure}
  \centerline{
  \includegraphics[width=0.35\textwidth]{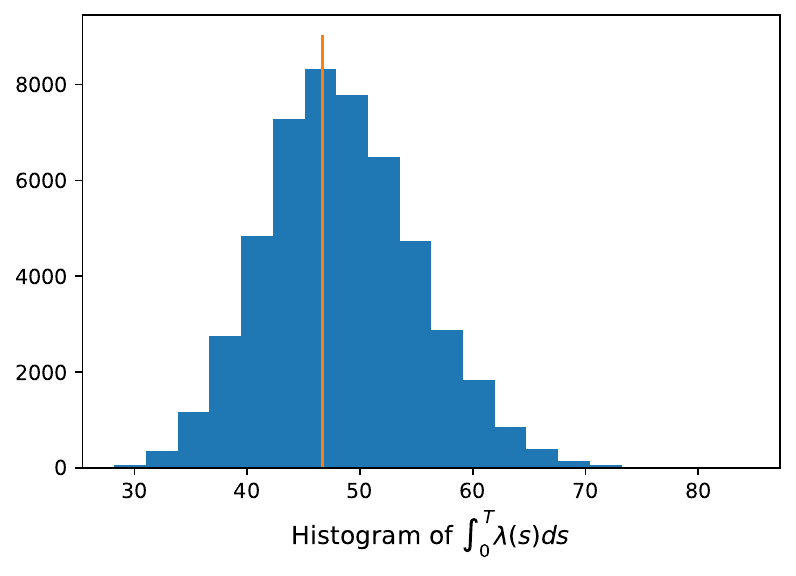}
   \includegraphics[width=0.35\textwidth]{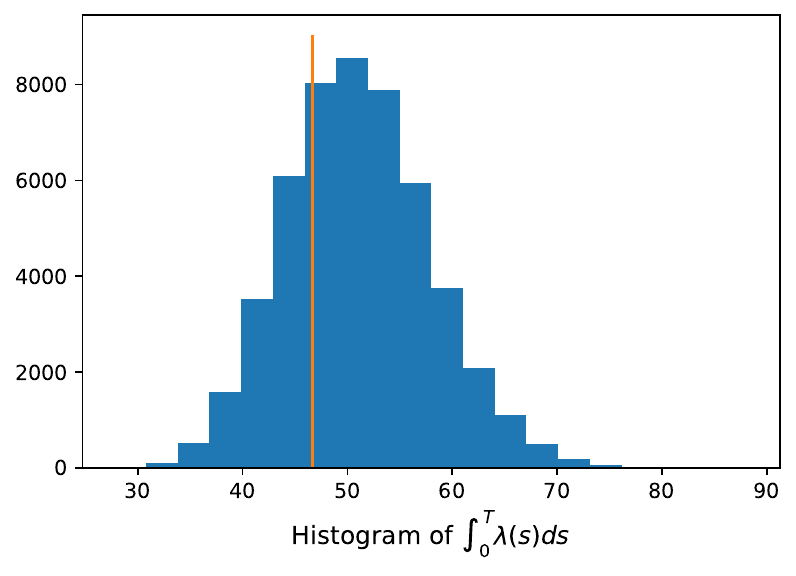}
   \includegraphics[width=0.35\textwidth]{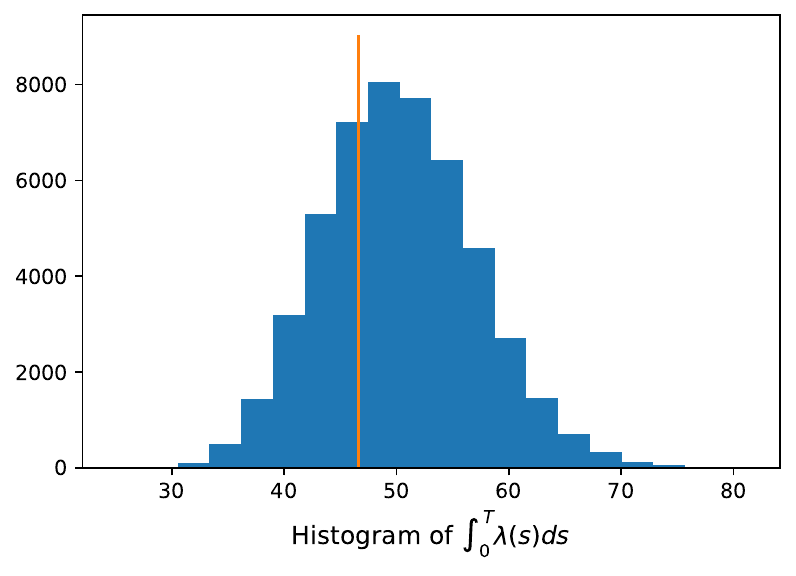}
  
  }
  \centerline{
  \includegraphics[width=0.35\textwidth]{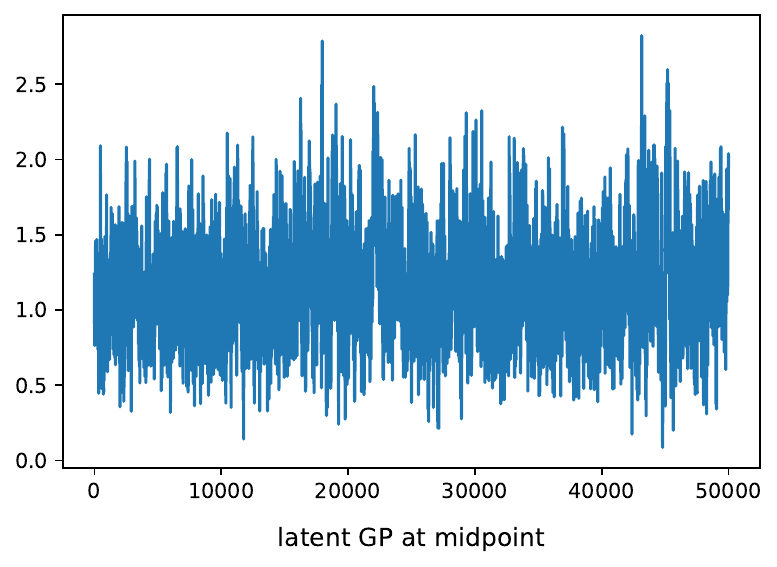}
  \includegraphics[width=0.35\textwidth]{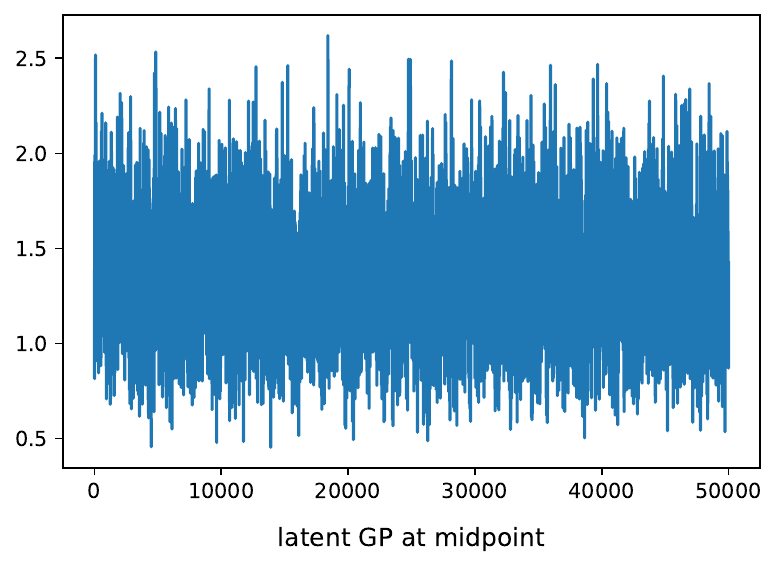}
   \includegraphics[width=0.35\textwidth]{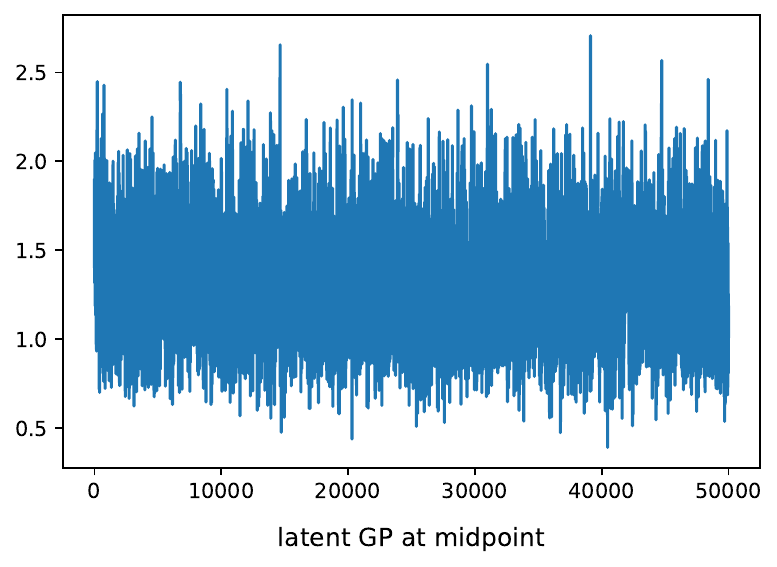}
  }
   \captionsetup{font=footnotesize}
    \vspace{\baselineskip}
  \caption{
  MCMC traceplots for posterior samples of latent GP at midpoint along with histograms of the cumulative intensity function for the same simulated dataset from $\lambda_1(s)$ as that in top row of \cref{fig:syn1}. On the top row, the red vertical line in each histogram represents the ground truth value for $\int_0^T \lambda_1(s)ds$. The three columns from left to right correspond to results using methods RI-BM, RI-MAP and RI-MLE respectively. }\label{fig:syn1traceplots}
\end{figure}

\begin{table}
\captionsetup{font=footnotesize}
\caption{ Performance comparison on simulations with intensities $\lambda_1(s)$,  $2\lambda_1(s)$ and $3\lambda_1(s)$. The last three columns present quantities in the format: 0.50 quantile (0.025 quantile, 0.25 quantile, 0.75 quantile, 0.975 quantile). Bold is the best among methods with the same kernel and red is the best among all methods.}
\label{table:syn1full}

\vspace{1\baselineskip} 

\centering
\begin{adjustbox}{width=\textwidth} 
\begin{tabular}{p{1.2cm}p{1.8cm}p{5.6cm}p{4.9cm}p{4.6cm}} 
\toprule[1.5pt]
\textbf{Expected Event Counts} & \textbf{Methods} & \textbf{SSE at Observations}& \textbf{Coverage at Observations}&\textbf{Credible Interval Width based on Observations} \\
\midrule
\multirow{7}{*}{46.65} &RI-BM&\textbf{3.21}\;\;\;(\textbf{1.04},\;\;\textbf{2.25},\;\;\;\;\;\textcolor{red}{4.16},\;\;\;\; \textcolor{red}{8.34})&98\%\;\;(83\%,   95\%,\;\;\;\textcolor{red}{100\%},\;\textcolor{red}{100\%})&1.32 (1.15, 1.25, 1.39, 1.61)\\
&SGCP-BM&4.07\;\;\;(1.10,\;\;2.86,\;\;\;\;\;7.57,\;\;\;23.20)&\textcolor{red}{100\%}(80\%, \textcolor{red}{100\%},\;\;\textcolor{red}{100\%}, \textcolor{red}{100\%})&1.60 (1.30, 1.46, 1.73, 2.09)\\
&INLA-BM &4.19\;\;\;(1.49,\;\;3.18,\;\;\;\;\;6.32,\;\;\;15.93) & 
\textcolor{red}{100\%}(\textcolor{red}{88\%}, \textcolor{red}{100\%},\;\;\textcolor{red}{100\%},  \textcolor{red}{100\%})&2.02 (1.33, 1.69, 2.30,  2.90)\\ 
\cmidrule(r){2-5}
& RI-MAP& 4.23\;\;\;(1.53,\;\;3.08,\;\;\;\;\;5.46,\;\;\;10.46) &\textbf{90\%}\;(\textbf{47\%},\;\;\textbf{77\%},\;\;\textcolor{red}{100\%},\;\textcolor{red}{100\%})&0.96 (0.63, 0.80, 1.18, 1.56)\\
& SGCP-MAP& \textbf{3.63}\;\;\;(\textbf{1.21},\;\;\textbf{3.04},\;\;\;\;\;\textbf{4.59},\;\;\;\;\;\textbf{9.85}) & 88\%\;\;(\textbf{47\%},\;\;74\%,\;\;\textcolor{red}{100\%},\;\textcolor{red}{100\%})&0.94 (0.64, 0.80, 1.11, 1.40)\\
\cmidrule(r){2-5}
&RI-MLE& 3.55\;\;\;(\textcolor{red}{0.84},\;\;2.40,\;\;\;\;\;\textbf{4.41},\;\;\;\;\;\textbf{9.64})
 &
\textcolor{red}{100\%}(\textbf{74\%},\;\;94\%,\;\;\textcolor{red}{100\%}, \textcolor{red}{100\%})&1.19 (1.05, 1.14, 1.24, 1.31)\\
&SGCP-MLE&\textcolor{red}{ 3.17}\;\;\;(0.98,\;\;\textcolor{red}{2.08},\;\;\;\;\;4.52,\;\;\;10.31)
 &
\textcolor{red}{100\%}(70\%,\;\;\;\textbf{96\%},\;\textcolor{red}{100\%},\;\textcolor{red}{100\%})&1.15 (1.01, 1.09, 1.18, 1.26)\\
\midrule[1.5pt]
 \multirow{3}{4em}{93.3} &RI-BM&\textbf{17.70}\;(6.95,\;\;13.72,\;\;\;\textcolor{red}{23.45},\;\;\textcolor{red}{44.76})&99\%\;\;(79\%,\;\;94\%,\;\;\textcolor{red}{100\%},\;\textcolor{red}{100\%})&1.83 (1.62, 1.74, 1.93, 2.15)\\
&SGCP-BM&19.42\;(\textbf{4.06},\;\;\textbf{13.21},\;\;\;33.42,\;\;91.06)&\textcolor{red}{100\%} (87\%, \textcolor{red}{100\%},  \textcolor{red}{100\%}, \textcolor{red}{100\%})&2.58 (2.12, 2.43, 2.75, 3.26)\\ 
&INLA-BM & 21.77\;(6.72,\;\;15.63,\;\;\;35.66,\;\;98.56)& \textcolor{red}{100\%} (\textcolor{red}{89\%}, \textcolor{red}{100\%}, \textcolor{red}{100\%}, \textcolor{red}{100\%})& 3.06 (2.36, 2.83, 3.42, 4.11)  \\ \cmidrule(r){2-5}
& RI-MAP& 24.20\;(5.40,\;\;17.56,\;\;\;29.65,\;\;60.73)
 &93\%\;\;\;(53\%,\;\;77\%,\;\;\textcolor{red}{100\%}, \textcolor{red}{100\%})&1.61 (1.04, 1.44, 1.79, 2.28)\\
\cmidrule(r){2-5} 
 &RI-MLE& \textcolor{red}{14.66}\;(\textcolor{red}{3.53},\;\;\textcolor{red}{10.76},\;\;\;25.12,\;\;64.12)
 &
\textcolor{red}{100\%}\;\;(71\%,\;\;98\%,\;\;\textcolor{red}{100\%},\;\textcolor{red}{100\%})&1.81 (1.66, 1.76, 1.88, 1.97)\\\midrule[1.5pt]
 
 \multirow{3}{4em}{139.95} &RI-BM& \textbf{49.32} (24.13, \textbf{38.53},\;\;\textcolor{red}{63.82},\; \textcolor{red}{121.38}) &99\%\;\;\;(69\%,\;\;\;93\%,\;\;\textcolor{red}{100\%},\;\textcolor{red}{100\%})&2.32 (1.91, 2.20, 2.50, 2.63)\\
 &INLA-BM & 65.31 (\textbf{20.04}, 42.13, 101.72,  277.04)&\textcolor{red}{100\%}\;\;(\textcolor{red}{88\%},\;\;\textcolor{red}{99\%},\;\;\textcolor{red}{100\%},\;\textcolor{red}{100\%}) &4.14 (3.03, 3.77, 4.51, 5.29)\\
 \cmidrule(r){2-5}
 &RI-MAP& 63.46 (13.30, 35.73, \;89.16,\; 157.57)
  &94\%\;\;\;(49\%,\;\;\;77\%,\;\;\textcolor{red}{100\%},\;\textcolor{red}{100\%})&2.09 (1.60, 1.88, 2.34, 2.86)\\
\cmidrule(r){2-5}
 &RI-MLE&\textcolor{red}{37.26}\;\;(\textcolor{red}{7.92},\;\;\textcolor{red}{20.94},\;\;66.00,\;132.17)
 &
\textcolor{red}{100\%}\;\;(73\%,\;\;96\%,\;\;\textcolor{red}{100\%},\;\;\textcolor{red}{100\%})&2.35 (2.11, 2.28, 2.41, 2.54)\\
\midrule[1.5pt]
\textbf{Expected Event Counts} & \textbf{Methods} &\textbf{SSE at 100 Grids} & \textbf{Coverage at 100 Grids}&\textbf{Credible Interval Width based on 100 Grids}\\
\midrule
 \multirow{3}{4em}{93.3} 
 &RI-BM&\textbf{16.20} (7.39,\;\;\textbf{12.58}, \textbf{20.97},\;\;\textcolor{red}{43.88})&99\%\;\;(80\%,\;\;96\%,\;\;\textcolor{red}{100\%},  \textcolor{red}{100\%})&1.68 (1.46, 1.57, 1.74, 1.95)\\
&SGCP-BM&17.84  (\textbf{5.78},\;\;13.54,    28.46,\;\;71.43)&\textcolor{red}{100\%} (\textcolor{red}{86\%}, \textcolor{red}{100\%}, \textcolor{red}{100\%}, \textcolor{red}{100\%})&2.14 (1.82, 2.02, 2.28, 2.70)\\
&INLA-BM & 19.95  (7.30,\;\;15.12,  27.41,\;\;51.68)&
\textcolor{red}{100\%}\;(85\%,\;\;99\%,\;\;\textcolor{red}{100\%},  \textcolor{red}{100\%})&2.33 (1.89, 2.21, 2.51, 2.75)\\\cmidrule(r){2-5}
 & RI-MAP&20.60   (5.46,\;\;15.11,    26.20,\;\;48.31)
 &92\%\;\;\;(53\%,\;\;80\%,\;\;\textcolor{red}{100\%},  \textcolor{red}{100\%})&1.39 (0.87, 1.25, 1.55, 1.93)\\
\cmidrule(r){2-5}
 &RI-MLE& \textcolor{red}{13.35}\;(\textcolor{red}{3.13},\;\;\textcolor{red}{9.95},\;\;\textcolor{red}{19.11},\;\;48.04)
 &\textcolor{red}{100\%}\;\;(74\%,  95\%,\;\;\textcolor{red}{100\%},  \textcolor{red}{100\%})&1.55 (1.39, 1.50, 1.60, 1.67)\\
\midrule[1.5pt]
  \multirow{3}{4em}{139.95} 
  &RI-BM& \textbf{27.91} (\textbf{12.00}, \textbf{19.76}, \textbf{37.28}, \textbf{67.12}) &99\%\;\;\;(78\%,\;\;93\%,\;\;\textcolor{red}{100\%}, \textcolor{red}{100\%})&2.07 (1.74, 1.95, 2.21, 2.40)\\
   &INLA-BM &35.65  (12.55,   24.77,   53.81,    93.92)&\textcolor{red}{100\%}\;\;(\textcolor{red}{88\%},\;\;\textcolor{red}{99\%},\;\textcolor{red}{100\%},  \textcolor{red}{100\%})&3.15 (2.46, 2.91, 3.34, 3.70)\\\cmidrule(r){2-5}
  &RI-MAP& 35.86\;\;(8.12,\;\;18.92,  48.65,    85.82)
  &94\%\;\;\;(49\%,\;\;81\%,\;\;\textcolor{red}{100\%},  \textcolor{red}{100\%})&1.79 (1.39, 1.62, 1.94, 2.47)\\\cmidrule(r){2-5}
&RI-MLE& \textcolor{red}{21.44}\;\;(\textcolor{red}{4.26},\;\;\textcolor{red}{11.74},  \textcolor{red}{37.21},    \textcolor{red}{65.77})
 &
\textcolor{red}{100\%}\;\;(79\%,\;\;96\%,\;\;\textcolor{red}{100\%},  \textcolor{red}{100\%})&1.98 (1.83, 1.92, 2.03, 2.12)\\

\bottomrule[1.5pt]
\end{tabular}
\end{adjustbox}

\vspace{1\baselineskip} 

\end{table}

\begin{figure}
  \centerline{
  \includegraphics[width=0.35\textwidth]{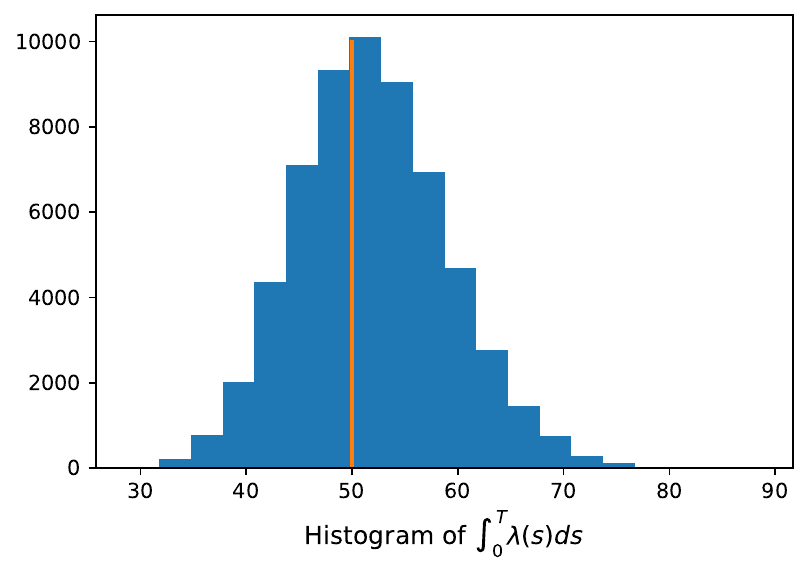}
   \includegraphics[width=0.35\textwidth]{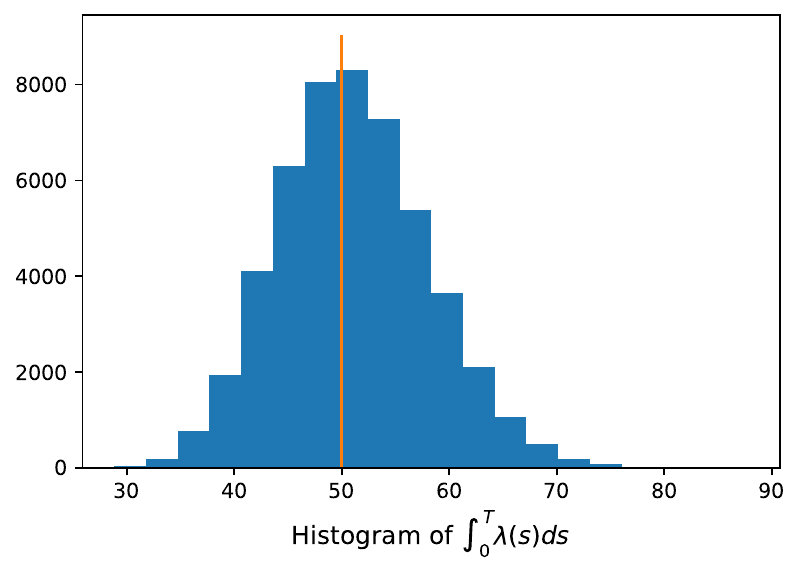}
   \includegraphics[width=0.35\textwidth]{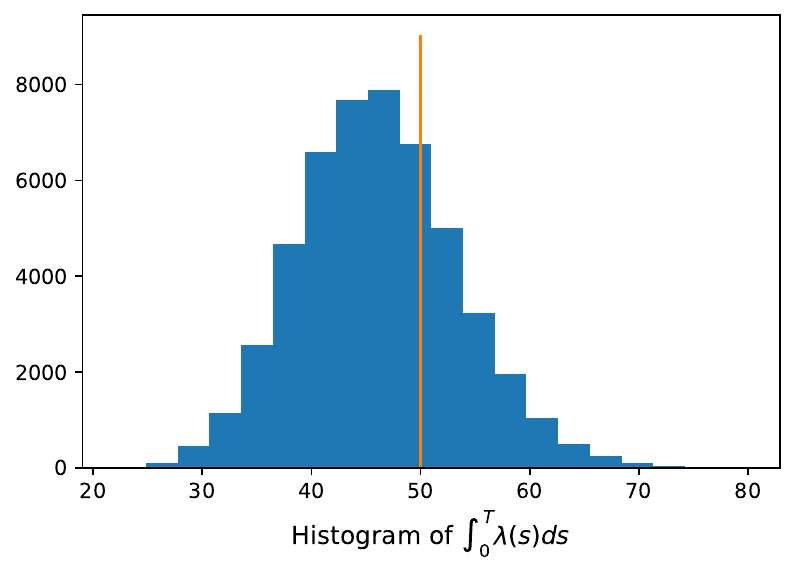}
  
  }
  \centerline{
  \includegraphics[width=0.35\textwidth]{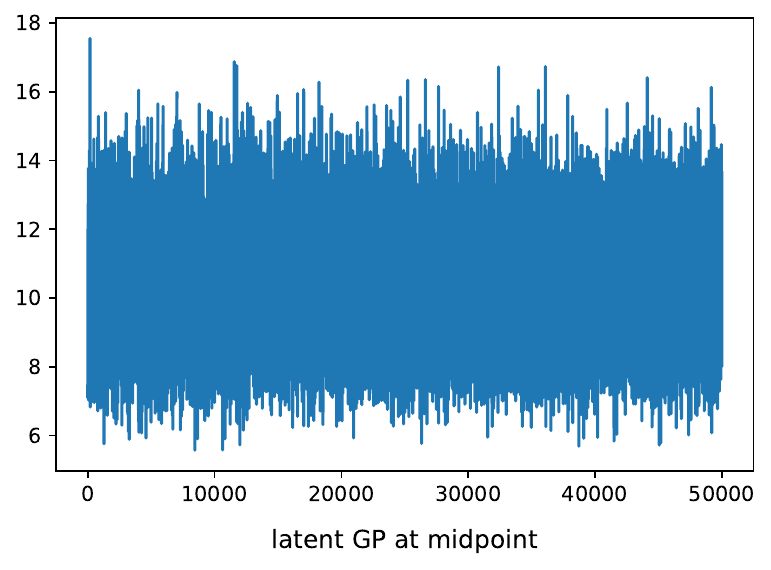}
  \includegraphics[width=0.35\textwidth]{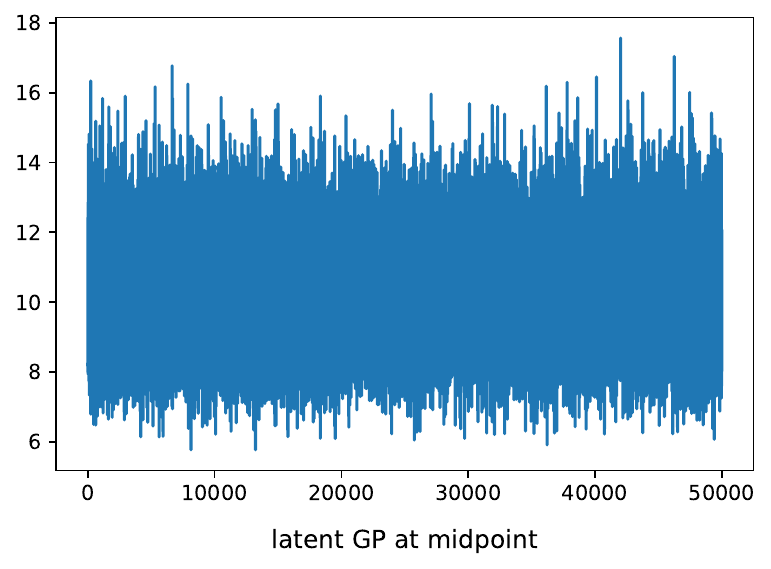}
   \includegraphics[width=0.35\textwidth]{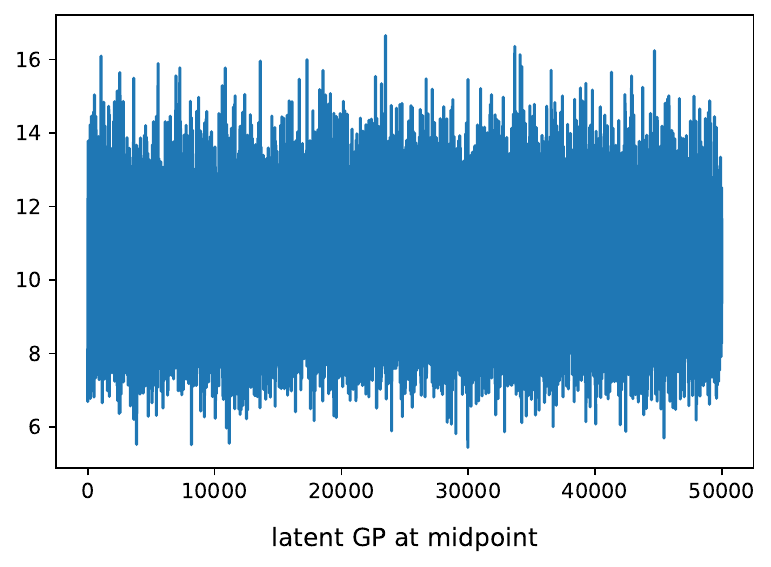}
  }
   \captionsetup{font=footnotesize}
     \vspace{\baselineskip}
  \caption{ MCMC traceplots for posterior samples of latent GP at midpoint along with histograms of the cumulative intensity function for the same simulated dataset from $\lambda_2(s)$ as that in bottom row of \cref{fig:syn1}. On the top row, the red vertical line in each histogram represents the ground truth value for $\int_0^T \lambda_2(s)ds$. The three columns from left to right correspond to results using methods RI-BM, RI-MAP and RI-MLE respectively. }\label{fig:syn2traceplots}
\end{figure}
\begin{table}
\footnotesize
\captionsetup{font=footnotesize}
\caption{ Running time for estimating intensities $2\lambda_1(s)$ and $3\lambda_1(s)$.}
\label{table:syn1timefull}

\vspace{1\baselineskip} 

\centering
\begin{adjustbox}{width=0.9\textwidth} 
\begin{tabular}{p{4.2cm}p{1.8cm}p{6.3cm}} 
\toprule[1.5pt]
\textbf{ Expected Event Counts} & \textbf{Methods}& \textbf{Average Time $\pm$ Standard Deviation}\\ [0.5ex]
 \hline
 \multirow{3}{4em}{93.3} &RI-BM&3.68 $\pm$ 1.52 s\\ 
 &SGCP-BM&234392.72 $\pm$ 21852.78 s\\
 &INLA-BM &
4.81 $\pm$ 0.85 s\\
\cmidrule(r){2-3}
 & RI-MAP& 3.67 $\pm$ 0.29 s\\
  \cmidrule(r){2-3}
&RI-MLE&
3.81 $\pm$ 0.39 s\\

\midrule[1.5pt]
  \multirow{3}{4em}{139.95} &RI-BM&3.22$ \pm$ 0.60 s \\
   &INLA-BM & 
4.79 $\pm$ 0.12 s\\
 \cmidrule(r){2-3}
  & RI-MAP&  3.97 $\pm$ 0.30 s\\
   \cmidrule(r){2-3}
 &RI-MLE&
4.22 $\pm$  0.28 s\\

\bottomrule[1.5pt]
\end{tabular}
\end{adjustbox}


\end{table}
\begin{table}
\captionsetup{font=footnotesize}
\caption{  Performance comparison on simulations with intensities $\lambda_2(s)$,  $2\lambda_2(s)$ and $3\lambda_2(s)$. The last three columns present quantities in the format: 0.50 quantile (0.025 quantile, 0.25 quantile, 0.75 quantile, 0.975 quantile). Bold is the best among methods with the same kernel and red is the best among all methods.}
\label{table:syn2full}

\vspace{1\baselineskip} 

\centering
\begin{adjustbox}{width=\textwidth} 
\begin{tabular}{p{1.2cm}p{1.8cm}p{6.3cm}p{4.9cm}p{4.7cm}} 
\toprule[1.5pt]
\textbf{Expected Event Counts} & \textbf{Methods} & \textbf{SSE at Observations}& \textbf{Coverage at Observations}&\textbf{Credible Interval Width based on Observations} \\
\midrule
  \multirow{3}{4em}{50} &RI-BM &\textcolor{red}{40.57}\;\;\;(\textcolor{red}{1.55},\;\;\textcolor{red}{11.68},\;\textcolor{red}{107.33},\;\;\;\textcolor{red}{411.35}) &\textcolor{red}{100\%}\;(\textcolor{red}{100\%},\textcolor{red}{100\%},\textcolor{red}{100\%},\textcolor{red}{100\%}) &6.28\;\;\;(5.43,\;\;5.95,\;\;6.73,\;10.22) \\%
&SGCP-BM &112.96 (16.37,  55.14,  236.44,   1.45e+3)&\textcolor{red}{100\%}\;(77\%,\;\;\textcolor{red}{100\%},\textcolor{red}{100\%},\textcolor{red}{100\%})&9.10\;\;\;(7.77,\;\;8.45,\;10.08,12.55)\\
&INLA-BM &146.70    (35.54,   86.12,  251.31,  1.18e+3)& \textcolor{red}{100\%}\;(94\%,\;\;\textcolor{red}{100\%},\textcolor{red}{100\%},\textcolor{red}{100\%})&12.71\;(10.53,11.79,14.13,19.60)\\ 
\cmidrule(r){2-5}
  & RI-MAP& \textbf{63.88}\;\;\;(\textbf{0.10},\;\;\textbf{18.20},\;\;\textbf{152.75},\;\;\;\textbf{683.79})
 &\textcolor{red}{100\%}\;(\textbf{68\%}, \textcolor{red}{100\%},\textcolor{red}{100\%},\textcolor{red}{100\%})&5.60\;\;\;(4.81,\;\;5.26,\;\;5.87,\;\;9.70)\\
 & SGCP-MAP& 65.93\;\;\;(\textbf{0.10},\;\;18.21,\;\;213.01,\;1.24e+3)
 &\textcolor{red}{100\%}\;(39\%,\;\;\textcolor{red}{100\%},\textcolor{red}{100\%},\textcolor{red}{100\%})&5.82\;\;\;(4.78,\;\;5.43,\;\;6.29,\;10.31)\\\cmidrule(r){2-5}
& RI-MLE& \textbf{43.51}\;\;\;(\textbf{0.45},\;\;\textbf{13.69},\;\;\textbf{133.90},\;\;\;548.31)
 &
\textcolor{red}{100\%}\;(\textcolor{red}{100\%},\textcolor{red}{100\%},\textcolor{red}{100\%},\textcolor{red}{100\%})&5.84\;\;\;(4.85,\;\;5.41,\;\;6.84,\;13.60)\\
&SGCP-MLE& 53.96\;\;\;(0.46,\;\;17.70,\;\;136.95,\;\;\;\textbf{503.35})
 &
\textcolor{red}{100\%}\;(0\%,\;\;\;\;\textcolor{red}{100\%},\textcolor{red}{100\%},\textcolor{red}{100\%})&5.65\;\;\;(4.77,\;\;5.29,\;\;5.95,\;\;6.71)\\

\midrule[1.5pt]
  \multirow{3}{4em}{100} &RI-BM& \textbf{282.11}\;(\textbf{2.88},\;\;\;\;\textcolor{red}{72.87},\;\;\;\textcolor{red}{585.80},\;\;\textcolor{red}{1.63e+3})&\textcolor{red}{100\%}\;(62\%,\;\textcolor{red}{100\%},\;\textcolor{red}{100\%},\textcolor{red}{100\%})&9.10\;\;\;(7.68,\;\;8.52,\;\;9.77,\;12.31)\\
&SGCP-BM& 778.82  (83.47,\;\;267.56, 1.65e+3,\;5.42e+3)&\textcolor{red}{100\%}\;(73\%,\;\textcolor{red}{100\%},\;\textcolor{red}{100\%},\textcolor{red}{100\%})&14.72(11.62,13.37,16.01,   21.40)\\
 &INLA-BM&801.15  (218.49,450.33, 1.40e+3,\;4.35e+3)&\textcolor{red}{100\%}\;(\textbf{81\%},\;\textcolor{red}{100\%},\;\textcolor{red}{100\%},\textcolor{red}{100\%})&18.92(15.67,18.03,20.66,   24.26)   \\\cmidrule(r){2-5}
  & RI-MAP& \textcolor{red}{213.48}\;\;(\textcolor{red}{2.01},\;\;\;82.61,\;\;\;693.65,\;\;2.35e+3)
 &\textcolor{red}{100\%}\;(0\%,\;\;\;\textcolor{red}{100\%},\;\textcolor{red}{100\%},\textcolor{red}{100\%})&7.75\;\;\;(6.96,\;\;7.46,\;\;8.16,\;\;10.61)\\\cmidrule(r){2-5}
&RI-MLE&677.97\;\;(8.73,\;\;143.01,\;4.38e+3,\;3.47e+5)
&
\textcolor{red}{100\%}\;(0\%,\;\;\;\textcolor{red}{100\%},\;\textcolor{red}{100\%},\textcolor{red}{100\%})&8.15\;\;\;(6.43,\;\;7.51,\;36.28,109.61)\\

\midrule[1.5pt]
 
 \multirow{3}{4em}{150} &RI-BM&\textbf{1.19e+3}(\textbf{10.9},\;\;\;\;\;\textcolor{red}{205.51},\;\; \textbf{2.44e+3}, \textbf{6.29e+3})&\textcolor{red}{100\%}\;(0\%,\;\;\;\textcolor{red}{100\%},\;\textcolor{red}{100\%},\textcolor{red}{100\%})&12.06\;(9.75,    11.11,13.60,\;\;16.50)\\
 &INLA-BM&2.16e+3 (547.67, 1.30e+3, 3.70e+3, 1.00e+4) &\textcolor{red}{100\%}\;(\textbf{84\%},\;\textcolor{red}{100\%},\;\textcolor{red}{100\%},\textcolor{red}{100\%})&24.53\;(20.45,23.13,26.05,   30.33) \\\cmidrule(r){2-5}
 & RI-MAP& \textcolor{red}{605.16}\;\;\;(\textcolor{red}{5.14},\;\;\;\;218.68,\;\;\textcolor{red}{1.52e+3},  \textcolor{red}{4.97e+3}) 
  &\textcolor{red}{100\%}\;(0\%,\;\;\;\textcolor{red}{100\%},\;\textcolor{red}{100\%},\textcolor{red}{100\%})
  &9.42\;\;\;(8.57,\;\;9.20,\;\;9.85,\;\;11.00)\\\cmidrule(r){2-5}
&RI-MLE&4.21e+5 (354.95, 1.36e+5, 2.62e+6, 3.90e+7)
 &
0\%\;\;\;\;(0\%,\;\;\;\;0\%,\;\;\;\;48\%,\;\;100\%)&36.90 (8.63, 16.19,110.35,742.42)\\
\midrule[1.5pt]
\textbf{Expected Event Counts} & \textbf{Methods} &\textbf{SSE at 100 Grids} & \textbf{Coverage at 100 Grids}&\textbf{Credible Interval Width based on 100 Grids}\\\midrule
  \multirow{3}{4em}{100} &RI-BM&\textbf{298.82}  (\textbf{2.66},\;\;\;\;\;\textcolor{red}{75.56},\;\;\;\;\textcolor{red}{571.92},\;\;\;\textcolor{red}{1.47e+3})& \textcolor{red}{100\%}\;(57\%,  \textcolor{red}{100\%},\;\textcolor{red}{100\%},\textcolor{red}{100\%})&9.05\;\;\;(7.68,\;\;8.49,\;\;9.73,\;12.25) \\
&SGCP-BM&884.86  (89.05,\;\;\;303.59,\;\;1.47e+3,\;\;4.31e+3)& \textcolor{red}{100\%}\;(80\%,  \textcolor{red}{100\%},\;\textcolor{red}{100\%},\textcolor{red}{100\%})&14.94(11.87,13.60,16.19,  21.88)\\
&INLA-BM &850.71  (257.30,\;\;512.17,\;1.48e+3,\;\;2.98e+3)&\textcolor{red}{100\%}\;(\textbf{88\%}, \textcolor{red}{100\%},\;\textcolor{red}{100\%},\textcolor{red}{100\%})&18.15(15.21,17.43,19.73,  21.98)\\
\cmidrule(r){2-5}
  
  &RI-MAP& \textcolor{red}{213.53}\;\;(\textcolor{red}{2.03},\;\;\;\;\;79.55,\;\;\;675.76,\;\;\;2.31e+3)
 &\textcolor{red}{100\%}\;(0\%,\;\;\;\textcolor{red}{100\%},\;\textcolor{red}{100\%},\textcolor{red}{100\%})&7.75\;\;\;(6.96,\;\;7.46,\;\;8.16,\;10.37)\\\cmidrule(r){2-5}
&RI-MLE&396.97\;\;(5.32,\;\;\;\;131.65,\;\;\;1.49e+3,\;\;2.38e+5)
&
\textcolor{red}{100\%}\;(0\%,\;\;\;\textcolor{red}{100\%},\;\textcolor{red}{100\%},\textcolor{red}{100\%})&8.15\;\;\;(6.43,\;\;7.51,37.86,111.93)\\
 
\midrule[1.5pt]
  \multirow{3}{4em}{150}&RI-BM&\textbf{830.65}\;\;\;(\textbf{6.77},\;\;\;\;\textcolor{red}{138.58},\;\;\textbf{1.74e+3},\;\textcolor{red}{4.33e+3})&\textcolor{red}{100\%}\;(0\%,\;\;\;\textcolor{red}{100\%},\;\textcolor{red}{100\%},\textcolor{red}{100\%})&12.04  (9.74, 11.09,13.52,   16.11) \\
   &INLA-BM &1.61e+3 (378.05,\;947.39,\;\;2.61e+3,  6.19e+3)&\textcolor{red}{100\%}\;(\textbf{82\%},\;\textcolor{red}{100\%},\;\textcolor{red}{100\%},\textcolor{red}{100\%})&23.59(20.09,22.42,24.99,   28.39)  \\
\cmidrule(r){2-5}
  & RI-MAP&
   \textcolor{red}{434.85}\;\;\;(\textcolor{red}{3.39},\;\;\;\;151.44,\;\;\;\textcolor{red}{999.74},\;\;4.73e+3)
  &\textcolor{red}{100\%}\;(0\%,\;\;\;\textcolor{red}{100\%},\;\textcolor{red}{100\%},\textcolor{red}{100\%})&9.42\;\;(8.57,\;\;9.20,\;\;9.85,\;\;11.04)\\\cmidrule(r){2-5}
&RI-MLE&1.94e+5 (234.9,\;\;6.97e+4,
1.69e+6, 1.47e+7)
 &
0\%\;\;\;\;(0\%,\;\;\;\;\;0\%,\;\;\;\;96\%,\;\textcolor{red}{100\%})&36.90(8.63,16.24,117.17,703.80)\\

\bottomrule[1.5pt]
\end{tabular}
\end{adjustbox}

\vspace{1\baselineskip} 

\end{table}

\begin{table}
\footnotesize
\captionsetup{font=footnotesize}
\caption{  Running time for estimating intensities $2\lambda_2(s)$ and $3\lambda_2(s)$.}
\label{table:syn2timefull}

\vspace{1\baselineskip} 

\centering
\begin{adjustbox}{width=0.9\textwidth} 
\begin{tabular}{p{4.2cm}p{1.8cm}p{6.3cm}} 
\toprule[1.5pt]
\textbf{ Expected Event Counts} & \textbf{Methods}& \textbf{Average Time $\pm$ Standard Deviation}\\ [0.5ex]
 \hline
\multirow{3}{4em}{100} &RI-BM&1.83 $\pm$ 0.36 s\\
&SGCP-BM&43811.94 $\pm$ 9316.78 s\\
 &INLA-BM &
4.47 $\pm$ 0.20 s\\\cmidrule(r){2-3}
& RI-MAP& 3.02$\pm$ 0.24 s\\\cmidrule(r){2-3}&RI-MLE&
3.82 $\pm$ 0.54 s\\

\midrule[1.5pt]
   \multirow{3}{4em}{150} &RI-BM&2.45 $\pm$  0.89 s \\
 &INLA-BM & 
4.34 $\pm$ 0.07 s\\\cmidrule(r){2-3}

& RI-MAP& 3.27 $\pm$ 0.18 s\\\cmidrule(r){2-3}
 &RI-MLE&
5.28 $\pm$ 0.61 s\\
\bottomrule[1.5pt]
\end{tabular}
\end{adjustbox}

\vspace{1\baselineskip} 

\end{table}
\subsection{ Earthquakes in Japan $\&$ redwoods}
~\cref{fig:earthappend,fig:redwoodsappend} present traceplots of latent GP at midpoint and histograms of the cumulative intensity function for real examples in \cref{sec:earthquake} and \cref{sec:realexm3}.

\begin{figure}
  \centerline{
  \includegraphics[width=0.4\textwidth]{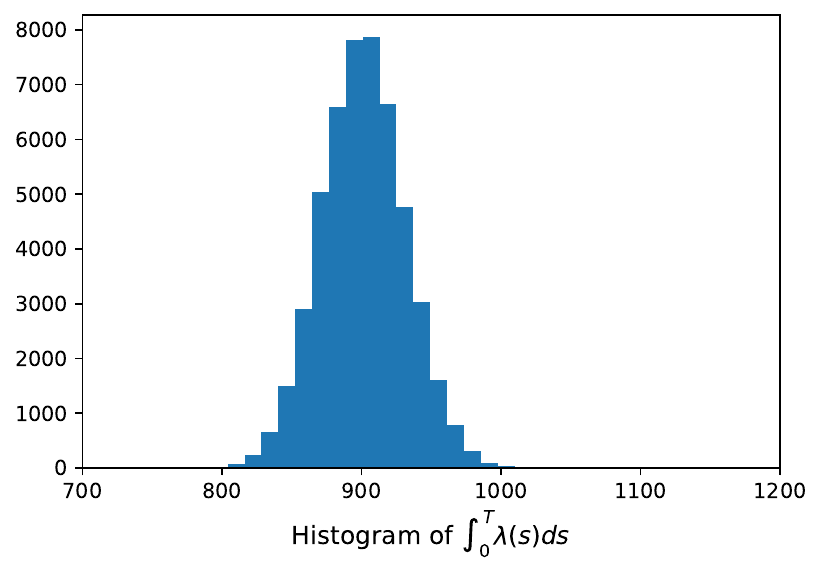}
  
   \includegraphics[width=0.4\textwidth]{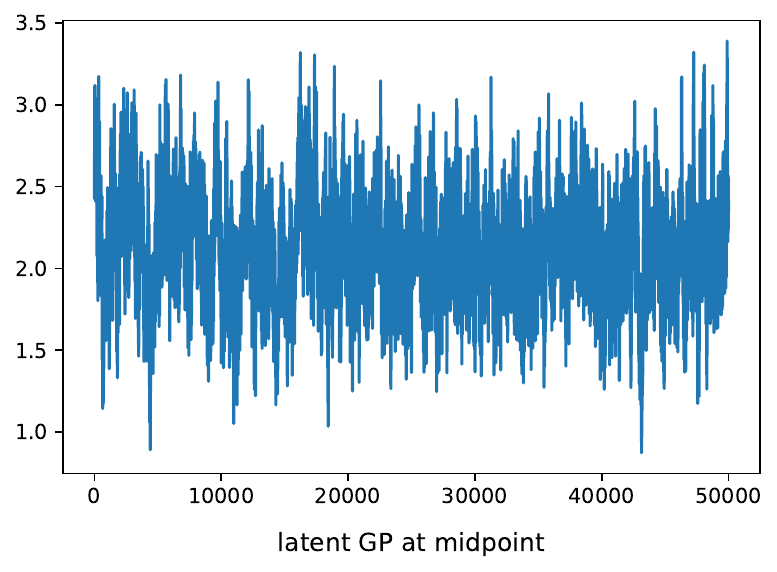}
  }
 \centerline{
  \includegraphics[width=0.4\textwidth]{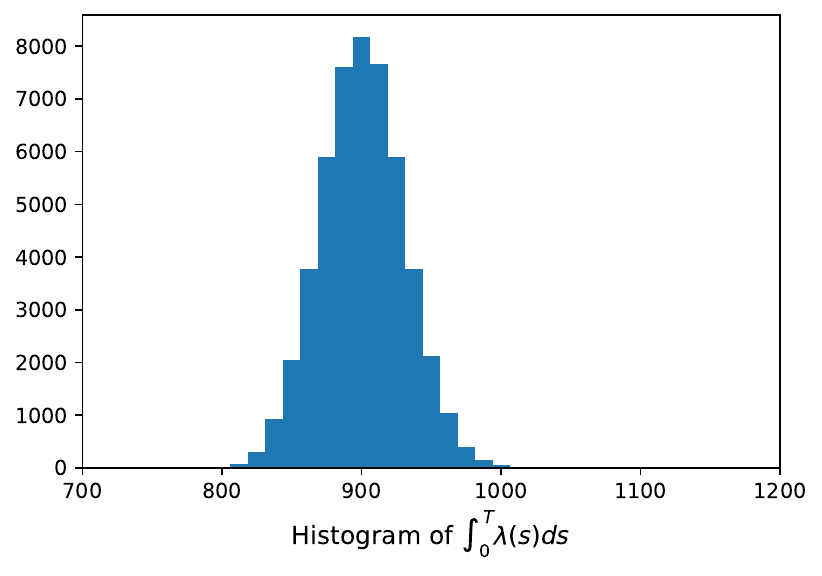}
  
   \includegraphics[width=0.4\textwidth]{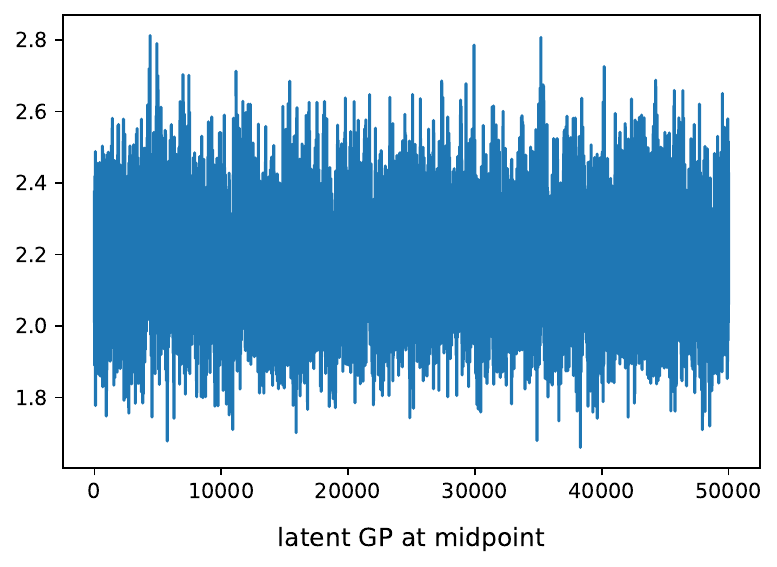}
  }
  \captionsetup{font=footnotesize}
  \vspace{\baselineskip}
  \caption{MCMC traceplots for posterior samples of latent GP at midpoint along with histograms of the cumulative intensity function for the recurrent earthquake data in \cref{sec:earthquake}. Top row corresponds to results using our proposed method RI-BM and bottom row corresponds to results using our proposed method RI-MAP.}\label{fig:earthappend}
\end{figure}

\begin{figure}
  \centerline{
  \includegraphics[width=0.4\textwidth]{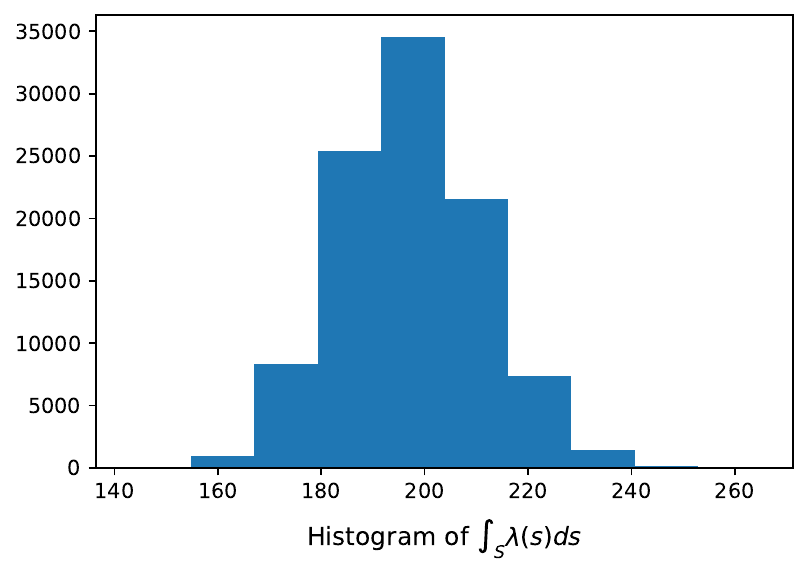}
  
   \includegraphics[width=0.4\textwidth]{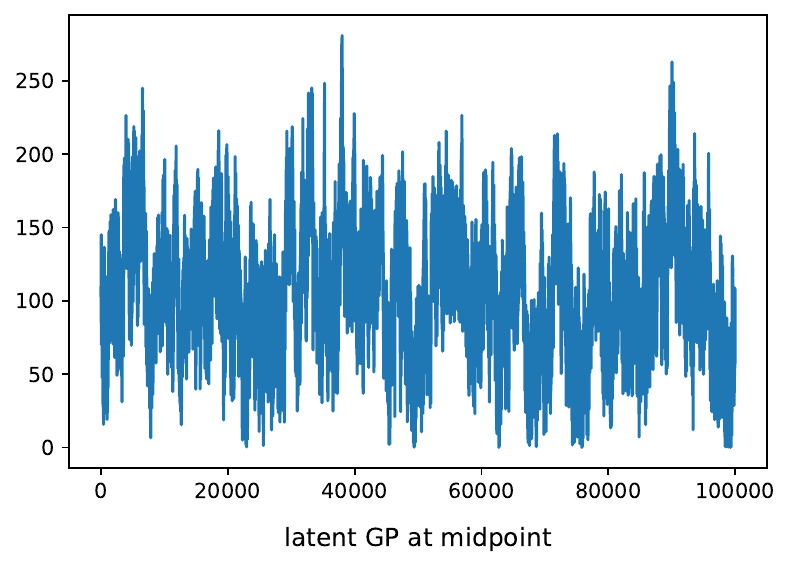}
  }
 \centerline{
  \includegraphics[width=0.4\textwidth]{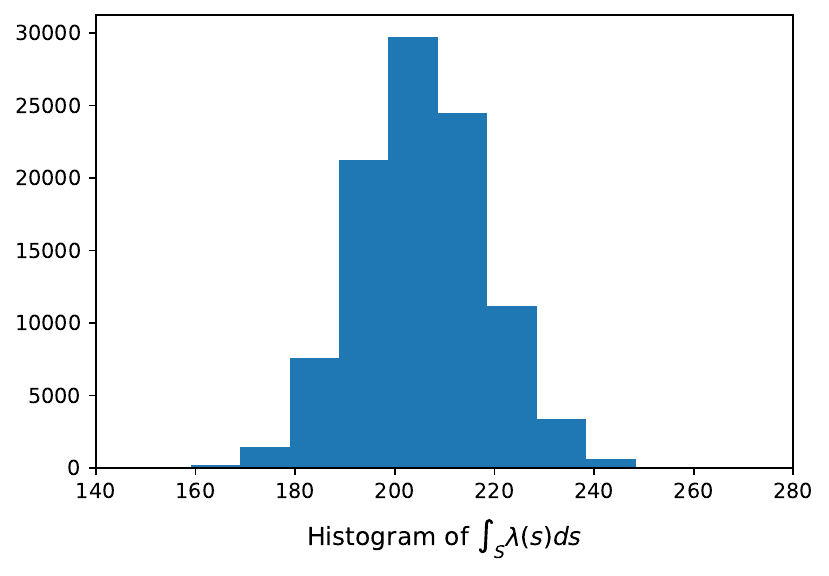}
  
   \includegraphics[width=0.4\textwidth]{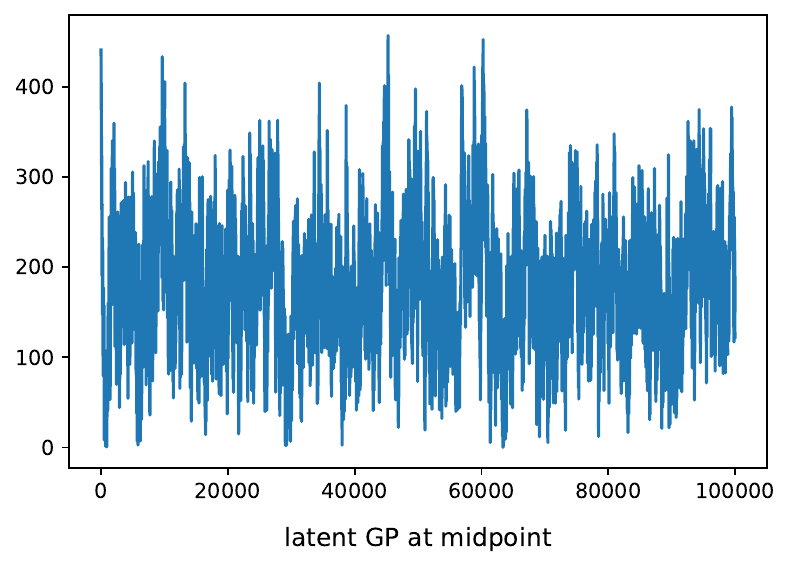}
  }
  \captionsetup{font=footnotesize}
   \vspace{\baselineskip}
  \caption{MCMC traceplots for posterior samples of latent GP at midpoint along with histograms of the cumulative intensity function for the redwoods data in \cref{sec:realexm3}. Top row corresponds to results using our proposed method RI-BM and bottom row corresponds to results using our proposed method RI-MAP.}\label{fig:redwoodsappend}
\end{figure}

\end{document}